\documentclass[prodmode]{acmsmall}

\usepackage{natbib}
\usepackage{amsmath,amssymb,epsfig,url,graphicx,xspace,stmaryrd}
\usepackage[ruled]{algorithm2e}
\usepackage[noend]{algorithmic-hacked}
\newcommand{\TAB}{\hspace{5mm}}

\newcommand{\initOneLiners}{%
 	\setlength{\itemsep}{0pt}
	\setlength{\parsep }{0pt}
  	\setlength{\topsep }{0pt}     	
}
\newenvironment{OneLiners}[1][\ensuremath{\bullet}]
    {\begin{list}
        {#1}
        {\initOneLiners}}
    {\end{list}}

\newcommand{\OMIT}[1]{}

\renewcommand{\Re}{\mathbb{R}}
\newcommand{\N}{\mathbb{N}}
\newcommand{\E} {\operatornamewithlimits{\ensuremath{\mathbb{E}}}} 
\newcommand{\argmax}{\operatornamewithlimits{argmax}}

\newcommand{\LDOTS}{\, ,\ \ldots\ ,}     

\newcommand{\xhdr}[1]{\vspace{2mm}\noindent{\bf #1}}
\newcommand{\refeq}[1]{Equation~(\ref{#1})}

\newcommand{\mA} {\ensuremath{\mathcal{A}}} 
\newcommand{\mP} {\ensuremath{\mathcal{P}}} 

\newcommand{\rM}{w} 
\newcommand{\rA}{s} 
\newcommand{\rP}{\bar{s}} 
\newcommand{\rN}{r} 
\newcommand{\rT}{q} 
\newcommand{\rY}{x} 
\newcommand{\rZ}{y} 

\newcommand{\half}{\frac{1}{2}}

\newcommand{\cout}[1]{}

\newcommand{\tA}{\widetilde{\mA}} 
\newcommand{\tP}{\widetilde{\mP}} 

\newcommand{\CanonProc}{canonical self-resampling procedure\xspace}
\newcommand{\FirstAlg}{\text{Algorithm~\ref{alg:canon-proc}}}
\newcommand{\SecondAlg}{{\mbox{\sc OneShot}}}
\newcommand{\mba}{x}
\newcommand{\mbb}{y}

\newcommand{\AAi}[1]{\mA_i(#1)} 
\newcommand{\PPi}[1]{\mP_i(#1)} 

\newcommand{\eAAi}[1]{\mA_i(#1;\rM, \rN)} 
\newcommand{\ePPi}[1]{\mP_i(#1;\rM, \rN)} 

\newcommand{\A}[1]{\mA(#1;\rA, \rN)} 
\newcommand{\Aii}[1]{\mA_i(#1;\rA, \rN)} 
\newcommand{\Pii}[1]{\mP_i(#1;\rP, \rN)} 

\newcommand{\noPositiveTransfers}{no-positive-transfers}

\newcommand{\outcomes}{\mathcal{O}}
\newcommand{\types}{\mathcal{T}}
\newcommand{\type}{x}

\newcommand{\UCB}{\ensuremath{\tt UCB1}}
\newcommand{\EXP}{\ensuremath{\tt EXP3}}
\newcommand{\bmax}{\ensuremath{b_\text{max}}}
\newcommand{\Rex}{\ensuremath{R_{\texttt{ex}}}} 

\newcommand{\MABprob}{MAB mechanism design problem}
\newcommand{\newUCB}{{\tt NewCB}}

\newcommand{\OPT}{\mathtt{OPT}}
\newcommand{\SW}{\mathtt{SW}} 

\usepackage[suppress]{color-edits}  
\addauthor{bk}{magenta} 
\addauthor{as}{red}      
\addauthor{mb}{blue}    

\acmVolume{0}
\acmNumber{0}
\acmArticle{0}
\acmYear{2014}
\acmMonth{0}

\newcommand{\qedhere}{}

\begin{document}

\markboth{M. Babaioff et al.}{Truthful Mechanisms  with Implicit Payment Computation}

\title{Truthful Mechanisms  with Implicit Payment Computation}

\date{First version: April 2010\\This version: November 2015}

\OMIT{ 
\footnote{This is a full version of a conference paper published in \emph{11th ACM Conf. on Electronic Commerce (EC)}, 2010.
Apart from the revised presentation, this version is updated to reflect the follow-up work and the current status of open questions. Moreover, the multi-parameter extension in Section~\ref{sec:multiParam} is new (i.e., this extension was not present in the conference version).}
} 

\author{MOSHE BABAIOFF
    \affil{Microsoft Research, Herzeliya, Israel.}
ROBERT D. KLEINBERG
\affil{Computer Science Department, Cornell University, Ithaca, NY, USA.}
ALEKSANDRS SLIVKINS
\affil{Microsoft Research, New York, NY, USA.}}


\begin{abstract}
It is widely believed that 
computing payments needed to induce truthful bidding
is somehow harder than simply computing the allocation.
We show that the opposite is true: 
creating a randomized truthful mechanism
is essentially as easy as a \emph{single} call to
a monotone allocation rule.
Our main result is a general procedure
to take a monotone allocation rule for a single-parameter domain
and transform it (via
a black-box reduction) into a randomized mechanism
that is truthful in expectation and 
individually rational for every realization.
The mechanism implements the same outcome as the
original allocation rule with probability
arbitrarily close to $1$,
and requires evaluating that allocation rule only
once. We also provide an extension of this result
to multi-parameter domains and cycle-monotone
allocation rules, under mild star-convexity
and non-negativity hypotheses on the type space
and allocation rule, respectively.

Because our reduction is simple, versatile, and
general,  it has many applications to mechanism
design problems in which re-evaluating the allocation rule is either burdensome or informationally impossible.
Applying our result to the
multi-armed bandit problem, we obtain truthful randomized
mechanisms whose regret matches the information-theoretic
lower bound up to logarithmic factors, even though prior
work showed this is impossible for truthful deterministic
mechanisms.  We also present applications to offline
mechanism design, showing that randomization can circumvent
a communication complexity lower bound for deterministic payments
computation, and that it can also be used to create truthful shortest path auctions that
approximate the welfare of the VCG allocation arbitrarily well, while
having the same running time complexity as Dijkstra's
algorithm.
\end{abstract}

\hyphenation{Compu-ters}
\hyphenation{Compu-tation}




\begin{bottomstuff}
This is a merged and revised version of the conference papers \cite{Transform-ec10-conf,BKS2-ec13} that have appeared in the \emph{ACM Conf. on Electronic Commerce (ACM EC)} in 2010 and 2013, respectively. This paper contains all results from \cite{Transform-ec10-conf} and the main result from \cite{BKS2-ec13} (in Section~\ref{sec:multiParam}). This version is updated to reflect the current status of the follow-up work and open questions.
\vspace{2mm}

Parts of this research have been done while R. Kleinberg was a Consulting Researcher at Microsoft Research Silicon Valley. He was also supported by NSF Awards CCF-0643934 and AF-0910940, an Alfred P. Sloan Foundation Fellowship, and a Microsoft Research New Faculty Fellowship.
\vspace{2mm}
\end{bottomstuff}

\maketitle

\newpage
\section{Introduction}
\label{sec:intro}

Algorithmic Mechanism Design studies the problem of implementing
the designer's goal under computational constraints.
Multiple hurdles stand in the way for such implementation.
Computing the desired outcome might be hard
(as in the case of combinatorial auctions) or truthful payments
implementing the goal
might not exist 
(as when exactly minimizing the make-span in machine scheduling~\cite{ArcherTardos}).
Even when payments that will generate the right incentives do exist,
finding such payments might be computationally costly or impossible
due to online constraints.

It is widely believed that
computing payments needed to induce truthful bidding
is somehow harder than simply computing the allocation.
For example, the formula for payments in a
VCG mechanism involves recomputing the allocation with
one agent removed in order to determine
that agent's payment; this seemingly increases the
required amount of computation by a factor of $n+1$,
where $n$ is the number of agents.
Likewise, for truthful single-parameter mechanisms
the formula for payments of a given agent includes
integrating the allocation rule over this agent's bid~\cite{Myerson,ArcherTardos}.
In some contexts
with incomplete observable information, such as online pay-per-click
auctions, computing these ``counterfactual allocations''
may actually be information-theoretically impossible.
This calls into question the mechanism designer's ability
to compute payments that make an allocation rule
truthful, even when such payment functions are known
to exist.  Rigorous lower bounds based on these observations
have been established for the communication complexity~\cite{BabaioffBS13} 
and regret~\cite{MechMAB-ec09,DevanurK09} of
truthful {\em deterministic} mechanisms.

In contrast to these negative results,
we show that the opposite is true for randomized single-parameter mechanisms
that are truthful-in-expectation:
computing the allocation and payments
is essentially as easy as a \emph{single} call to the allocation rule.
This allows for positive results
that circumvent the lower bounds for deterministic
mechanisms cited earlier.

\subsection{Single-parameter mechanisms}

We consider an arbitrary single-parameter domain. The paradigmatic example is an auction that allocates items between agents whose utility is linear in the number of items they receive. The private information of each agent is expressed by a single parameter: her value per item.%
\footnote{In a general single-parameter domain, the allocation rule selects an outcome from some arbitrary collection of feasible outcomes. Each agent has her own type of ``good'', and for each agent there is an arbitrary, publicly known mapping from feasible outcomes to a real-valued amount of the corresponding good. The agent's utility is linear in this amount;  the value per unit amount of good is her private information.}
Each agent submits a bid, then the mechanism performs the allocation and charges payments. A mechanism is called ``truthful" if each agent maximizes her utility by submitting her true value per item. The allocation rule in a truthful mechanism is called ``truthfully implementable''. It is known that an allocation rule is truthfully implementable if and only if it is ``monotone": increasing one agent's bid while keeping all other bids the same does not decrease this agent's allocation~\cite{Myerson,ArcherTardos}. A similar property holds for randomized mechanisms and truthfulness-in-expectation.

\xhdr{Our contributions.}
Our main result is a general procedure to take any monotone-in-expectation allocation rule $\mA$ and transform it into a randomized mechanism that is truthful-in-expectation, implements the same outcome as $\mA$ with probability arbitrarily close to $1$, and requires evaluating that allocation rule only once. (We refer to this procedure as the \emph{generic transformation}.) The allocation rule $\mA$ is accessed only as a function call, so our result applies even if $\mA$ is an online algorithm. Moreover, for each realization of randomness an agent never loses by participating in the mechanism and bidding truthfully; thus the agents are protected from undesirable random deviations.

\OMIT{\footnote{Such mechanisms are called \emph{individually rational}.}}

We make a distinction between randomness in the mechanism and randomness in ``nature'': the environment that the mechanism interacts with. Randomness in nature is subject to modeling assumptions and hence is less ``reliable''; moreover, agents' beliefs about nature may be different from the mechanism's. On the other hand, randomness in the mechanism is fully controlled by the mechanism. Therefore it is desirable to design mechanisms that are truthful in a stronger sense: in expectation over the mechanism's random seed, for every realization of randomness in nature; we will call such mechanisms \emph{ex-post truthful}. It is easy to see from~\cite{Myerson,ArcherTardos} that in any ex-post truthful mechanism the allocation rule must satisfy ex-post monotonicity (which is defined similarly to ex-post truthfulness). In the generic transformation described above, if the original allocation rule $\mA$ is ex-post monotone then the resulting randomized mechanism is ex-post truthful.

Similarly, our result extends to Bayesian incentive-compatibility: if $\mA$ is monotone in expectation with respect to a Bayesian prior over other agents' bids, then the mechanism is truthful in expectation over this prior.

Our generic transformation is particularly useful for mechanism design problems in which re-evaluating the allocation rule is either burdensome or information-theoretically impossible.

\subsection{Bandit mechanisms}

A leading problem for which only a single call to the allocation rule can be evaluated is the multi-armed bandit (MAB) mechanism design problem~\cite{MechMAB-ec09,DevanurK09}. In this problem information about the state of the world is dynamically revealed during the allocation; the particular information that is revealed depends on the prior choices of the allocation, and in turn may impact the future choices. Simulating the allocation rule on different inputs may therefore require information that was not observed on the actual run. This ``informational obstacle" (insufficient observable information) is a crucial obstacle for deterministic ex-post truthful MAB mechanisms; it is used in~\cite{MechMAB-ec09} to derive that the appropriate payments cannot be computed unless the allocation rule is very ``na\"ive" (and therefore suboptimal).

To put more context, MAB mechanisms are motivated by online pay-per-click ad auctions, and were suggested in~\cite{MechMAB-ec09,DevanurK09} as a simple model which combines strategic bidding by agents and online learning by the mechanism. Each agent has a single ad that she wants to display to users, and derives utility only if her ad is clicked. The value per click is her private information. The allocation rule proceeds in rounds: in each round the mechanism allocates one ad to be shown to a user and observes whether this ad was clicked. The click probabilities (also known as ``click-through rates", or \emph{CTRs}) are unknown to the mechanism, and need to be estimated during the run of the allocation rule. All bids are submitted before the allocation starts, and all payments are assigned after it ends.

MAB mechanisms are related to MAB \emph{algorithms}: the allocation rule is essentially an MAB algorithm whose ``rewards" are clicks weighted by the corresponding bids. Moreover, welfare of an MAB mechanism is precisely the same as the total reward of its allocation rule.%
\footnote{This is because payments cancel out: the total amount paid by the agents is equal to the total amount received by the mechanism.}
Therefore one could directly compare the performance of truthful MAB mechanisms with that of MAB algorithms; both can be quantified using \emph{regret}: the loss in welfare compared to the benchmark which always picks the best ad.

Following~\cite{MechMAB-ec09,DevanurK09}, we focus on the stochastic version of the problem, i.e. we assume that the CTRs do not change over time. Then the ``randomness in nature" corresponds to the random clicks, and ex-post truthfulness means truthfulness for every realization of the clicks (but in expectation over the randomness in the mechanism). Note that ex-post truthfulness is a very strong property which holds even if the clicks are chosen by an oblivious adversary. As discussed in~\cite{MechMAB-ec09,DevanurK09}, this property is highly desirable, compared to the weaker notion of ``truthfulness in expectation over clicks", even if the corresponding mechanism has regret guarantees that only apply to the stochastic setting.

\xhdr{Our contributions.}
Applying our generic transformation to the MAB problem we derive that the problem of designing truthful MAB mechanisms
reduces to the problem of designing monotone MAB allocation rules. Such a problem has not been previously studied in the rich literature on MAB.

Our main result in this direction is a randomized MAB mechanism that is ex-post truthful and has regret $O(T^{1/2})$ for the stochastic version. This upper bound on regret matches the information-theoretic lower bound for \emph{algorithms} in the same setting (i.e., the lower bound holds even in the absence of incentive constraints).  This stands in contrast to the lower bound of~\cite{MechMAB-ec09}, where it was shown that deterministic ex-post truthful MAB mechanisms must suffer a larger regret of $\Omega(T^{2/3})$.

On a technical level, we design a new MAB allocation rule that is ex-post monotone and has regret $O(T^{1/2})$ for the stochastic setting. (We use it to obtain a randomized ex-post truthful MAB mechanism with the same regret.) Moreover, we show that \UCB~\cite{bandits-ucb1} (and a number of similar MAB algorithms)
give rise to MAB allocations that are monotone in expectation over clicks, and therefore can be transformed to randomized MAB mechanisms that are truthful in the same sense and have optimal regret.

The new ex-post monotone MAB allocation rule is deterministic, which rigorously confirms the intuition from~\cite{MechMAB-ec09,DevanurK09} that the impossibility results for deterministic MAB mechanisms are caused by the ``informational obstacle" (insufficient observable information about clicks) rather than ex-post monotonicity.

\subsection{Other contributions}

\xhdr{Power of randomization.} As a by-product of our analysis of MAB mechanisms, we obtain an unconditional separation between the power of randomized vs. deterministic ex-post truthful mechanisms for welfare maximization, in the online setting. (The separation result is unconditional in the sense that it considers exactly the same setting for both classes of mechanisms.) This complements the result of \citeN{DD09}, which gives a separation between these two classes of mechanisms in the \emph{offline} setting, under a polynomial communication complexity constraint. It is worth noting that the separation in~\citeN{DD09} applies to a rather unnatural problem (two-player multi-unit auctions in which if at least one item is allocated, then all items are allocated and each player receives at least one item) whereas our separation result is for a natural problem: online pay-per-click ad auctions for a single slot, with unknown click-through rates.

For the objective of revenue maximization, separations between randomized and deterministic mechanisms have been known for much longer
\cite{Than,MV,j-stoc11-dobzi,j-soda10-bries} and are in some sense
less surprising. Randomization allows the mechanism to access a larger
set of possible allocations, i.e.\ the set of all probability distributions
over pure allocations, and in some cases this leads to greater revenue,
for example by permitting more fine-grained price discrimination
between agent types.
This is not the case for the objective of maximizing
welfare (because VCG mechanisms are deterministic and they
maximize welfare pointwise while obeying incentive constraints).
For welfare maximization, randomized mechanisms are sometimes more powerful than deterministic ones due to other reasons, such as computational power or informational limitations (as in the problem we study).

\xhdr{Offline mechanisms.}
Our main result also has implications for offline mechanism design.
Nisan and Ronen, in their seminal paper~\cite{NR01} which started the field of algorithmic mechanism design, cite the apparent $n$-fold computational overhead of computing VCG payments and pose the open question of whether payments can be computed faster than solving $n$ versions of the original problem, e.g.~for VCG path auctions.  Our result shows that the answer is affirmative, if one adopts the truthful-in-expectation solution concept and tolerates a mechanism that outputs an outcome whose welfare is a $(1+\epsilon)$-approximation to that of the VCG allocation, for arbitrarily small $\epsilon>0$. \citeN{BabaioffBS13} 
present a social choice function $f$ in an $n$-player single-parameter domain such that the deterministic communication complexity required for truthfully implementing $f$ exceeds that required for evaluating $f$ by a factor of $n$. Our result shows that no such lower bound holds when one considers randomized mechanisms, again allowing for a small amount of random error in the allocation.

\xhdr{Extension to multi-parameter mechanisms.} We extend our generic transformation from  single-parameter to multi-parameter mechanisms. It is known that a multi-parameter allocation rule is truthfully implementable if and only if it satisfies a property called ``cycle-monotonicity''. (This is a rather strong property which specializes to monotonicity in the single-parameter case.) Similar to the single-parameter case, we present a general procedure to take any cycle-monotone allocation rule $\mA$ and transform it into a randomized mechanism that is truthful-in-expectation, implements the same outcome as $\mA$ with probability arbitrarily close to $1$, and requires evaluating that allocation rule only once. The technical contribution here is that we find a reduction from the multi-parameter setting to the single-parameter case.

While much more general that our single-parameter transformation, this result may be more difficult to apply. This is because cycle-monotonicity is known to be a very restrictive property. However, the follow-up work already provides two applications, see Section~\ref{sec:followup} for details.

\subsection{Map of the paper}

This paper makes four high-level contributions: the generic transformation for single-parameter mechanisms (Sections~\ref{sec:generic} and~\ref{sec:welfare}), the two applications to off-line mechanism design (Section~\ref{sec:apps}), the results on MAB mechanisms (Section~\ref{sec:MAB}), and an extension to multi-parameter mechanisms (Section~\ref{sec:multiParam}). Presenting these results requires a significant amount of preliminaries on mechanisms design (Section~\ref{sec:prelims}), multi-armed bandits (Section~\ref{sec:MAB-prelims}), and multi-parameter mechanism (Section~\ref{sec:multiParam-prelims}). We conclude with open questions (Section~\ref{sec:open-questions}).

A considerable amount of work followed up on the initial conference publication \cite{Transform-ec10-conf} of this paper. This work is discussed in Section~\ref{sec:followup}.

\section{Related work and follow-up work}
\label{sec:related-work}

The characterization of truthful mechanisms for single-parameter
domains, given by \citeN{Myerson} for single-item auctions
and by \citeN{ArcherTardos} for a more general class
of single-parameter problems, states that a mechanism is truthful
if and only if its allocation rule is monotone and its payment
rule charges each agent its value for the realized outcome,
minus a correction term expressed as an integral over all
types lower than the agent's declared type.
Exact computation of this correction term may be
intractable, but \citeN{APTT04}
developed a clever workaround: one can use random sampling
to compute an unbiased estimator of the correction term,
at the cost of evaluating the allocation rule once more.
Thus, for $n$ agents, the allocation rule must be
evaluated $n+1$ times: once to determine the actual
allocation, and once more per agent to determine that agent's
payment.  Our generic transformation relies on a generalization of this
random sampling technique, but we show how to avoid recomputing
the allocation rule when determining each agent's payment,
by coupling payment generation with the allocation itself.

The question of whether computing payments is computationally harder than computing the allocation was raised by \citeN{NR01} in the context of VCG path auctions.  The most significant progress to date was the communication complexity lower bound of \citeN{BabaioffBS13} 
mentioned above.

\OMIT{, who constructed a single-parameter domain with a monotone social choice function
whose communication complexity is $n$ times less than the communication complexity of any deterministic incentive-compatible mechanism implementing this function.
In contrast, our results show that no such lower bound arises when one considers randomized truthful-in-expectation mechanisms with arbitrarily small probability of outputting the wrong allocation.}

Payment computation in online mechanism design is a central issue in the analysis of truthful MAB mechanisms in \citeN{MechMAB-ec09} and \citeN{DevanurK09}. The main result of~\cite{MechMAB-ec09} is a characterization of deterministic ex-post truthful mechanisms. It is more restrictive than the Myerson and Archer-Tardos characterization. The reason is that computing an agent's payment requires knowing how many clicks she would have received if she had submitted a lower bid value, which may require the mechanism to hypothetically go back into the past and allocate impressions to a different agent for the purpose of seeing whether a user would have clicked on that agent's advertisement. Such counterfactual information is typically impossible to obtain in an online setting.

\citeN{MechMAB-ec09} focus on welfare maximization. Using the above characterization, they prove that any deterministic ex-post truthful MAB mechanism must incur regret $\Omega(T^{2/3})$, whereas MAB algorithms for the same setting can achieve regret $O(T^{1/2})$. \citeN{DevanurK09} consider revenue maximization, and derive a similar $\Omega(T^{2/3})$ lower bound on loss of revenue compared to the VCG payments.\footnote{For revenue-maximizing MAB mechanisms, there is no clear comparison with the performance of MAB algorithms.}

\OMIT{
MAB mechanisms are motivated by online pay-per-click ad auctions with unknown or uncertain clickthrough rates.  In the absence of incentive constraints, the problem of learning click-through rates can be modeled as a MAB problem, and it is known that there are algorithms whose \emph{regret} (essentially, the loss in welfare due to not knowing the clickthrough rates at the outset of the auction) is $O(\sqrt{T})$ where $T$ is the number of impressions; moreover, this dependence on $T$ is
information-theoretically optimal.}

\emph{Dynamic auctions}~\cite{AtheySegal-econometrica13,DynPivot-econometrica10,DynAuctions-survey10} constitute another setting in which information is revealed ``dynamically" (over time). However, while in MAB auctions all information from the agents (the bids) is submitted only once and then information is revealed to the mechanism by the environment over time, in dynamic auctions the agents continuously observe private ``signals" from the environment and submit ``actions" to the mechanism. Accordingly, providing the right incentives becomes much more challenging. On the other hand, existing work has focused on a fully Bayesian setting with known priors on the signals, whereas all of our results do not rely on priors.

Finally, several recent papers have explored the theme of reductions in algorithmic mechanism design.
Unlike our work
which requires mechanisms to be truthful for
every realization of the agents' types,
these papers
focus on Bayesian settings and adopt Bayesian
incentive-compatibility as their solution concept.
A reduction converting any allocation rule into a Bayesian incentive-compatible
mechanism with approximately the same expected social
welfare was developed in \cite{HL10,BH11,HKM11}.
\citeN{CIL12} considered black-box reductions of mechanism design problems to algorithmic problems with the same objective, and demonstrated significant limitations of this approach.
The breakthrough results of
Cai et al.~\citeyearpar{CDW12,CDW13a,CDW13b} and \citeN{DW14} circumvented these limitations by instead reducing to algorithmic problems with a modified objective.
In particular, reductions from revenue-maximizing mechanisms to welfare-maximizing algorithms are presented in \cite{CDW12,CDW13a}, whereas
\citeN{CDW13b} and \citeN{DW14} present reductions for non-linear objective functions, such as makespan in scheduling.%
\footnote{All papers discussed in this paragraph, except \cite{HL10}, have appeared after the conference publication of this paper \cite{Transform-ec10-conf}.}

\subsection{Follow-up work (subsequent to \cite{Transform-ec10-conf})}
\label{sec:followup}

\OMIT{Let us describe the work~\cite{SingleCall-ec12,Gatti-ec12,Parkes-netecon12,Jain-sagt11,Huang-ExpMech12} which followed up on the conference publication of this paper.}

Our generic transformation exhibits high variability in payments, and includes an explicit tradeoff between the variability in payments and the loss in performance. Formally, variability can be expressed as variance, maximal absolute value, or (for positive types) maximal rebate. Performance can be expressed as welfare or revenue. \citeN{SingleCall-ec12} have proved this tradeoff to be optimal in a certain worst-case sense: our transformation achieves the optimal worst-case variance in payments for any given worst-case loss in performance, where the worst case is over all monotone allocation rules. Their result applies to any single-parameter domain and any of the above notions of variability and performance.

Our generic transformation is likely to be very useful in single-parameter settings which exhibit the ``informational obstacle" (insufficient observable information) such as the one found for deterministic MAB mechanisms. The follow-up work describes three additional settings. First, \citeN{SingleCall-ec12} observe that the same obstacle arises in offline pay-per-click ad auctions with multiple ad slots, where the CTRs have slot-specific multipliers. In conjunction with our generic transformation, an obvious welfare-maximizing allocation rule for that setting results in a truthful-in-expectation mechanism. Second, \citeN{Parkes-netecon12} describe a packet scheduling problem in a network router, where the ``informational obstacle" arises due to the potentially missing information about packet arrival times. (As they observe, this information may be missing not only because it is not observed by the router but also because the router simply does not have space to store it.) They design a monotone allocation rule for their setting, and use our generic transformation to convert it to a truthful-in-expectation mechanism. Third, \citeN{Gatti-ec12} consider an extension of MAB mechanisms to multiple ad slots. While they provide truthful mechanisms based on the simple MAB mechanism from~\cite{MechMAB-ec09,DevanurK09}, our generic transformation could give rise to more efficient truthful mechanisms.

\citeN{SingleCall-ec12} obtain a similar ``single-call reduction" (i.e., a reduction from allocation rules to truthful-in-expectation mechanisms which calls the allocation rule only once) for multi-parameter allocation rules that are \emph{maximal-in-distributional-range} (MIDR). MIDR allocation rules~\cite{DD09} pick a welfare-maximizing distribution over outcomes from some fixed collection of distributions; they are precisely the allocation rules for which VCG payments produce a truthful mechanism. This result is an independent work with respect to, and a special case of, the multi-parameter reduction in Section~\ref{sec:multiParam}.

The multi-parameter generic transformation in Section~\ref{sec:multiParam} has been used in two recent papers. First, \citeN{Jain-sagt11} used it to speed up the payment computation for a mechanism that allocates batch jobs in a cloud system. Second, \citeN{Huang-ExpMech12} used it to compute payments for their privacy-preserving procurement auction for spanning trees, which is based on the well-known ``exponential privacy mechanism" from prior work~\cite{McSherry-Talwar-focs07}.

\xhdr{Simplified payment computation.}
Our generic transformation is most useful if the allocation rule cannot be invoked more than once, as in ``bandit mechanisms'' or other examples provided in follow-up work. \citeN{Segal-personal-10} has observed that any truthful single-parameter mechanism can be implemented in a much simpler way, as long as \emph{two} calls to the allocation rule are allowed: one computes the allocation, and the other one generates random payments with the correct expectation. In the first call one uses the original bids. For the second call, one selects an agent uniformly at random, and uses the random sampling trick from \citeN{ArcherTardos} described above to compute the payment for this agent, and then scales the payment appropriately.%
\footnote{However, more work is needed for domains with negative agents' types, such as VCG shortest path auctions (see Section~\ref{sec:apps} for more details). In particular, one needs to carefully define the random sampling of the bid for payment computation, using a version of the argument in Section~\ref{sec:negative-types} to bound the loss in welfare.}

Also, a simpler generic transformation is possible if one settles for a weaker notion of Bayesian incentive-compatibility \cite{Hartline-book-draft}.

\section{Preliminaries}
\label{sec:prelims}

\xhdr{Single-parameter domains.}
We present the single parameter model for which we apply our procedure. The model is very similar to the model of Archer and Tardos~\cite{ArcherTardos}, yet it is slightly more general. We state the model is terms of values and not costs and allow the values to be both positive and negative. We also allow randomization by nature. All these changes are minor and do not change the fundamental characterization, yet are helpful to later derive our results.

Let $n$ be the number of agents and let $N=[n]$ be the set of agents.
Each agent $i\in N$ has some private {\em type}
consisting of a single parameter $\type_i\in \types_i$ that describes
the agent, and is known only to $i$, everything else is public knowledge.
We assume that the domain $\types_i$ is an open subset of $\Re$ which is an interval with positive length (possibly starting from $-\infty$ or going up to $\infty$).
Let $\types=\types_1\times \types_2\times ... \times \types_n$ denote the domain of types and let $t\in \types$ denote the vector of true types.

There is some set of {\em outcomes} $\outcomes$.
For single-parameter domains, agents evaluate outcomes in a particular way that we describe next.
For each agent $i\in N$ there is a function $a_i:\outcomes\rightarrow \Re_{+}$ specifying the {\em allocation to agent $i$}.
The {\em value} of an outcome $o\in \outcomes$ for an agent $i\in N$ with type $\type_i$ is $\type_i\cdot a_i(o)$.
The {\em utility} that agent $i\in N$ derives from outcome $o\in \outcomes$ when he is charged $p_i$ is quasi-linear:
$u_i = \type_i\cdot a_i(o) - p_i$.

For instance, consider the allocation of $k$ identical units of good to agents with additive valuations:
agent $i$ has a value of $\type_i$ {\em per unit}. An outcome $o$ specifies how many items each agent receives: $a_i(o)$ is the number of items $i$ receives. His valuation for that outcome is his value per-unit times the number of units he receives.

A (direct revelation) deterministic {\em mechanism} $\mathcal{M}$ consists of the pair $(\mA, \mP)$,
where $\mA:\types\rightarrow \outcomes$ is the {\em allocation rule} and $\mP:\types\rightarrow \Re^n$ is the {\em payment rule}, i.e. the vector of payment functions $\mP_i:\types\rightarrow \Re$ for each agent $i$.
Each agent is required to report a type $b_i\in \types_i$ to the mechanism, and $b_i$ is called the {\em bid} of agent $i$.
We denote the vector of bids by $b\in \types$.
The mechanism picks an outcome $\mA(b)$ and charges agent $i$ payment of $\mP_i(b)$.
The allocation for agent $i$ when the bids are $b$ is $\mA_i(b)= a_i(\mA(b))$ and he is charged $\mP_i(b)$.
Agent $i$'s utility when the agents bid $b\in \types$ and his type is $\type_i\in \types_i$ is
\begin{equation}\label{eq:util-det}
u_i(\type_i,b) =  \type_i\cdot \mA_i(b) - \mP_i(b)
\end{equation}


We also consider randomized mechanisms, which are distributions over deterministic mechanisms. For a randomized allocation rule $\mA_i(b)$ and $\mP_i(b)$ will denote the {\em expected} allocation and payment charged from agent $i$, when the bids are $b$. The expectation is taken over the randomness of the mechanism. Sometimes it will be helpful to explicitly consider the deterministic allocation and payment that is generated for specific random seed. in this case we use $\rM$ to denote the random seed and use $\mA_i(b;\rM)$ and $\mP_i(b;\rM)$ to denote allocation and payment when the seed is $\rM$.

There may be some outside randomization that influences the outcome and is not controlled by the mechanism, e.g. randomness in the realization of clicks in sponsored search auction. We call this randomization by {\em nature}.  With such randomization  $\mA_i(b)$ and $\mP_i(b)$ also encapsulate expectations over nature's randomization. Finally, we use the notation $\eAAi{b}$ and $\ePPi{b}$ to denote the allocation and payment charged from agent $i$, when the bids are $b$, the mechanism random seed is $\rM$ and nature's random seed is $\rN$.

\cout{
We allow the allocation and payment rules to be randomized and also allow nature to randomize (meaning that the outcome might depend on nature randomization) and assume that agent are risk neutral, each aims to
maximize his expected utility.
Formally, an allocation rule $A$ is a function that takes bids $b$ and a random seed $\rA$ and outputs a distribution over outcomes (with randomness coming from nature), $A(b;\rA)\in \Delta(\outcomes)$, where $\Delta(\outcomes)$ is the space of probability distributions over outcomes $\outcomes$. Alternatively, fixing $b$ and $\rA$ we can also think of the allocation as a deterministic function of nature randomness $\rN$.
We denote that allocation with bids $b$, allocation random seed $\rA$ and nature random seed $\rN$ as $\A{b}\in \outcomes$.
Now $\Aii{b}$ denotes the allocation to agent $i$, that is  $\Aii{b} = a_i(\A{b})$.
Similarly, the payment rule might be randomized with random seed $\rP$, and we can denote
the deterministic payment of agent $i$ when the bids are $b$, its random seed is $\rP$ and nature random seed is $\rN$ by $\Pii{b}\in \Re$. 
Agent $i$'s utility when the agents bid $b$ and his type is $\type_i$ is

\begin{equation}\label{eq:util}
u_i(\type_i,b) =  \E_{\rA,\rP,\rN}\left[ \type_i\cdot \Aii{b} - \Pii{b} \right]
\end{equation}

Note that the expectation is over all randomness (nature and rules).\footnote{One can also consider Bayesian settings in which agents have a prior over the types of others and try to maximize the expected utility, where expectation is taken also over the randomization generating the types of the others (rather than for a fixed $b_{-i}$). Our framework and results can be easily extended to this case, by capturing any such randomness similarly to the way we capture other randomness by nature and taking expectation of this randomness as well.}
} 

\xhdr{Allocation and Mechanism Properties.}
Let $b_{-i}$ denote the vector of bids of all agents but agent $i$. We can now write the vector of bids as $b=(b_{-i}, b_i)$. Similar notation will be used for other vectors.

We next list two central properties, truthfulness and individual rationality.
\begin{itemize}
\item
Mechanism $\mathcal{M}$ is {\em truthful} if for every agent $i$ truthful bidding is a {\em dominant strategy}: for every agent $i$, bidding $\type_i$ always maximizes her utility,
regardless of what the other agents bid.
Formally,
\begin{equation}\label{eq:truthful}
\type_i\cdot \AAi{b_{-i}, \type_i} - \PPi{b_{-i}, \type_i} \ge \type_i\cdot \AAi{b} - \PPi{b}
\end{equation}
holds 
for every agent $i\in N$, type $\type_i\in \types_i$, bids of others $b_{-i}\in \types_{-i}$ and bid $b_i\in \types_i$ of agent $i$.
\item Mechanism $\mathcal{M}$ is {\em individually rational (IR)} if an agent never receives negative utility by participating in the mechanism and bidding truthfully. Formally,
    \begin{equation}\label{eq:ir}
    \type_i\cdot \AAi{b_{-i}, \type_i} - \PPi{b_{-i}, \type_i} \ge 0
    \end{equation}
    holds 
    for every agent $i\in N$, type $\type_i\in \types_i$ and bids of others $b_{-i}\in \types_{-i}$.
\end{itemize}

It will be helpful to establish terminology for the case that the above hold not only in expectation but also for specific realizations. For example, we will say that a mechanism is {\em universally truthful} if \refeq{eq:truthful} holds not only in expectation over the mechanism's randomness, but rather for every realization of that randomness. In general, every property that we define is defined by some inequality, and if the inequality holds for every realization of the mechanism randomness we say that it holds {\em universally}, and if it holds for every realization of nature randomness we say that it holds {\em ex-post}. When we want to emphasize that the property holds only in expectation over the nature's randomness we say that it holds {\em stochastically}.

\OMIT{ 
As we have multiple sources of randomness it will be helpful to establish some terminology for quantifiers over the different randomness sources. A given property $p$ will be defined by some inequality.
For a given property $p$, we say that $p$ {\em holds} if the inequality holds in expectation over all sources of randomness.
The property holds {\em ex-post} if the inequality holds for every realization of nature randomness.
The property holds {\em universally} if the inequality holds for every realization of the mechanism randomness.
Thus, $p$ holds universally and ex-post if the inequality holds for every realization of both nature and the mechanism randomness.

Note that if nature does not randomize and the mechanism is deterministic the above properties hold both universally and ex-post.

If truthful bidding is a dominant strategy for all agents we say that the mechanism is {\em truthful}.

Mechanism $M$ is {\em individually rational (IR)} if for every agent $i$ with any type $\type_i$ and for every bids $b_{-i}$, an agent who bid truthfully never incurs a net loss, i.e. $u_i(\type_i, (b_{-i}, \type_i)) \ge 0$
} 

Note that in an individually rational mechanism an agent is ensured not to incur any loss {\em in expectation}.
That is rather unsatisfying as for some realizations the agent might suffer a huge loss.
It is more desirable to design mechanisms that are {\em universally ex-post individually rational}, that is
a truthful agent should incur no loss for every bids of the others and {\em every realization} of the random events (not only in expectation).


If all types are positive, then in addition to individual rationality it is desirable that all agents are charged a non-negative amount; this is known as the \emph{\noPositiveTransfers} property.
\OMIT{ 
If all types are positive, a stronger property is desirable (and satisfied by~\refeq{eq:myerson}): a mechanism is called \emph{normalized} if it is individually rational and moreover all agents are charged a non-negative amount.
} 
Finally, the \emph{welfare} of a truthful mechanism is defined to be the total utility
	$\sum_i \type_i \cdot \mA_i(t)$.

\xhdr{Characterization.}
The following characterization of truthful mechanisms, due to Archer and Tardos~\cite{ArcherTardos}, is almost identical to the characterization presented by Myerson~\cite{Myerson} for truthful mechanisms in the special case of single item auctions. The crucial property of an allocation that yields truthfulness is \emph{monotonicity}, defined as follows:

\begin{definition}\label{def:mon}
Allocation rule $\mA$ is {\em monotone} if for every agent $i\in N$, bids $b_{-i}\in \types_{-i}$ and two possible bids of $i$, $b_i\ge b^{-}_i$, we have
$\AAi{b_{-i}, b_i}\ge  \AAi{b_{-i}, b^{-}_i}$.
\end{definition}
Recall that monotonicity of an allocation rule is also defined universally and/or ex-post.

We next present the characterization of truthful mechanisms.  In the theorem statement, the
expression $\AAi{b_{-i},u}$ is interpreted to equal zero when $u \not\in \types_i.$

\begin{theorem}\label{thm:myerson}~\cite{Myerson,ArcherTardos}
Consider an arbitrary single-parameter domain. An allocation rule $\mA$ admits a payment rule $\mP$ such that the mechanism $(\mA,\mP)$ is truthful if and only if $\mA$ is monotone and moreover for each agent $i$ and bid vector $b$
it holds that
 $ \int_{-\infty}^{b_i} \AAi{b_{-i}, u} \,du < \infty$.
In this case the payment $\PPi{b}$ for each agent $i$ must satisfy
\begin{align}\label{eq:myerson}
\PPi{b}  =
 \mP_i^0(b_{-i}) +
 b_i\, \AAi{b_{-i}, b_i} -
 \textstyle{\int^{b_i}_{-\infty}}\,  \AAi{b_{-i}, u} \,du,
\end{align}
where $\mP_i^0(b_{-i})$ does not depend on $b_i$.
\end{theorem}

A mechanism is called \emph{normalized} if for each agent $i$ and every bid vector $b$, zero allocation implies a zero payment:
	$\AAi{b}=0 \Rightarrow \PPi{b}=0$.

\begin{corollary}\label{cor:myerson}
The truthful mechanism in Theorem~\ref{thm:myerson} is normalized if and only if $\mP_i^0(b_{-i}) \equiv 0$, in which case the mechanism is also individually rational and for positive-only types ($\types\subset \Re^n_+$) it moreover satisfies the no-positive-transfers property.
\end{corollary}

\OMIT{
\begin{theorem}\label{thm:myerson}~\cite{Myerson,ArcherTardos}
A monotone allocation rule $\mA$ admits a payment rule $\mP$ such that the mechanism $(\mA,\mP)$ is truthful and individually rational if and only if
for each agent $i$ and bid vector $b$
	$ \int_{-\infty}^{b_i} \AAi{b_{-i}, u} \,du < \infty$
In this case
\begin{align}\label{eq:myerson}
\PPi{b}  = b_i \AAi{b_{-i}, b_i} - \textstyle{\int^{b_i}_{-\infty}}\,  \AAi{b_{-i}, u} \,du.
\end{align}
\end{theorem}
}

Both Theorem~\ref{thm:myerson} and Corollary~\ref{cor:myerson} hold in the ``ex-post" sense (resp., ``universal" sense), if  $\mA_i$, $\mP_i$ and $\mP^0_i(b_{-i})$ are interpreted to mean their respective values for a specific random seed of nature (resp., mechanism). In Corollary~\ref{cor:myerson}, the mechanism is normalized in the same sense as it is truthful.

\section{The generic transformation for single-parameter domains}
\label{sec:generic}

This section presents a generic procedure which takes any monotone allocation rule for a single-parameter domain and creates a randomized truthful-in-expectation mechanism which attains the same outcome as the original allocation rule with high probability. The resulting mechanism uses the allocation rule as a ``black box,'' calls it only once, and allocates according to the this call. Henceforth, we will refer to this procedure as the \emph{generic transformation}.

Our main result --- the existence of the generic transformation with the desired properties --- can be stated informally as follows.

\begin{theorem}[Informal] \label{thm:main-informal}
Consider an arbitrary single-parameter domain with $n$ agents. Let $\mA$ be a monotone allocation rule for this domain. Then for each $\mu\in[0,1]$ there exists a truthful mechanism
	$\mathcal{M} = (\tA,\tP)$
with the following properties:
\begin{itemize}
\item $\mathcal{M}$ executes a single call to $\mA(\tilde{b})$ to compute the allocation, with a pre-processing step to compute the modified bid vector $\tilde{b}$, and a post-processing step to compute the payments. Both pre- and post-processing steps take $O(n)$ time and do not depend on $\mA$.

\item For any bid vector $b$ and any fixed random seed of nature allocations $\tA(b)$ and $\mA(b)$ are identical with probability at least $1-n\mu$.

\item $\mathcal{M}$ is universally ex-post individually rational. If all types are positive, then $\mathcal{M}$ is ex-post \noPositiveTransfers, and never pays any agent $i$ more than
	$b_i\cdot \mA_i(\mba)\cdot (\frac{1}{\mu}-1)$.
\OMIT{ 
Moreover, the modified types satisfy
	$\E[\rY_i(b_i;w_i)] = (1 - \frac{\mu}{2-\mu}) b_i$ for all $i$.
The welfare of $\mathcal{M}$ is at least
	$(1 - \frac{\mu}{2-\mu})$
times that of $\mA$.}
\end{itemize}
\end{theorem}

\newcommand{\GenericMech}{{\tt AllocToMech}}

Presenting the formal version of this result (Theorem~\ref{thm:maim-g}) requires defining the generic transformation. We begin with an informal description thereof. As evidenced
by~\refeq{eq:myerson}, the payment for agent $i$ is a
difference of two terms: the agent's reported utility
(i.e., the product of her bid and her allocation),
minus the integral of the allocation assigned to
every smaller bid value.  We charge the agent for
her reported utility, and we give her a random rebate
whose expectation equals the required integral.  When
integrating a function over a finite interval, an
unbiased estimator of the integral can be obtained
by sampling a uniformly random point of that
interval and evaluating the function at the sampled
point.  This idea was applied, in the context of
mechanism design, by \citeN{APTT04}.
Below, we show how to generalize the transformation
to allow for integrals over unbounded intervals,
as required by~\refeq{eq:myerson}.  Using this transformation
it is easy to transform any monotone
allocation rule into a randomized mechanism that is truthful
in expectation and only evaluates the
allocation rule $n+1$ times: once to determine
the actual allocation, and once more per agent to
obtain an unbiased estimate of that agent's payment.

Our main innovation is a transformation that uses the same random sampling trick, but only needs to evaluate the allocation rule once during the entire mechanism. (In other words, it does not require additional calls to the allocation rule to compute the payments.) Assume that a parameter $\mu \in (0,1)$ is given.  For every player, with probability $1-\mu$, we leave their bid unchanged; with probability $\mu$, we sample a smaller bid value. The allocation rule is invoked on these bids. An agent is always charged her reported value of the outcome, but if her bid was replaced with a smaller bid value then we refund her an amount equal to an unbiased estimator of the integral in~\refeq{eq:myerson}, scaled by $1/\mu$ to counterbalance the fact that the refund is only being applied with probability $\mu.$ A na\"{i}ve application of this plan suffers from the following defect: the random resampling of bids modifies the expected allocation vector, so we need to obtain an unbiased estimator of the integral of the \emph{modified} allocation rule.  However, if we change our sampling procedure to obtain such an estimate, then this modifies the allocation rule once again, so we will still be estimating the wrong integral!  What we need is a ``fixed point'' of this process of redefining the sampling procedure. Below, we give a definition of \emph{self-resampling
procedures} that satisfy the requisite fixed point property,
and we give two simple constructions of self-resampling
procedures.

A self-resampling procedure transforms the bid $b_i$ of a given agent $i$ bid into \emph{two} correlated random values $(x_i,y_i)$, where $x_i$ is the modified bid presented to the original allocation rule, and $y_i$ is used (together with the allocation itself) in computing the payment for this agent. More specifically, $y_i$ is needed to correctly normalize the unbiased estimator of the integral in~\refeq{eq:myerson} for the modified allocation rule, according to Theorem~\ref{thm:estimator} below.
For agents with positive types we define a simpler self-resampling procedure for which the unbiased estimator does not depend on $y_i$, and therefore, strictly speaking, the procedure only needs to output $x_i$ (more details can be found in Section~\ref{sec:simpler-transformation}).
However, we explicitly return the $y_i$ even for the positive types so as to be consistent with the general definitions and (perhaps more importantly) because we use it to define self-resampling procedures with general support (see Section~\ref{sec:canon-proc-more}).


\OMIT{ 
\BKnote{Per Moshe's request, I tried adding some intuition here.
After several attempts, I still feel that the sentence I added
does more harm than good; I feel that it will only confuse the
reader.  Here is the best sentence I could come up with.
{\bf In intuitive terms, we use two random values instead
of one because the second (larger) random value corresponds
to the resampled bid that is used for computing an unbiased
estimator of the integral, and the first (smaller) random
value results in an additional modification of the allocation rule, with the net effect that no further modification is
needed to correct for the modification of the allocation rule that takes place when a player's bid is replaced
with the second random value for pricing purposes.}
}} 

Thus, the formal description of our generic transformation consists
of three parts:
\begin{enumerate}
\item a method for estimating integrals by
evaluating the integrand at a randomly sampled point,
\item the definition and construction of self-resampling
procedures,
\item the generic transformation that
uses the previous two ingredients to convert any
monotone allocation rule into a truthful-in-expectation
randomized mechanism.
\end{enumerate}
We now specify the details of
each of these three parts.

\subsection{Estimating integrals via random sampling}
Let $I$ be a nonempty open interval in $\Re$ (possibly with infinite endpoints)
and let $g$
be a function defined on $I$.  Let us describe a procedure
for estimating the integral $\int_{I} g(z) \, dz$
by evaluating $g$ at a single randomly sampled point of $I$.
The procedure is well known; we describe it here for
the purpose of giving a self-contained exposition of our
algorithm.

\begin{theorem} \label{thm:estimator}
Let $F : I \rightarrow [0,1]$ be any strictly increasing
function that is differentiable and satisfies
	$\inf_{z\in I} F(z) = 0$ and $\sup_{z\in I} F(z) = 1$.
If $Y$ is a random variable with cumulative distribution
function $F$, then
\begin{align*}
    \int_{I} g(z) \, dz = \E\left[  \frac{g(Y)}{F'(Y)} \right].
\end{align*}
\end{theorem}

\begin{proof}
Since $\inf_{z\in I} F(z)=0$
and $\sup_{z\in I} F(z)=1$, it follows that the
random variable $Y$ is supported on the entire interval
$I$.  Our assumption that $F$ is differentiable
implies that $Y$ has a probability density function,
namely $F'(z).$  Thus, for any function $h$, the
expectation of $h(Y)$ is given by
$\int_I h(z) F'(z) \, dz.$  Applying
this formula to the function $h(z) = g(z)/F'(z)$
one obtains the theorem.
\end{proof}

\subsection{Self-resampling procedures}
The basic ingredient of our generic transformation is a
procedure for taking a bid $b_i$ and a random seed $\rM_i$,
and producing two random numbers $\rY_i(b_i;\rM_i), \, \rZ_i(b_i;\rM_i).$
The mechanism will use $\{ \rY_i(b_i;\rM_i)\}_{i\in N}$ for determining the allocation
and additionally $\rZ_i(b_i;\rM_i)$ for determining the payment it charges agent $i$.  To prove
that the mechanism is truthful in expectation we will require
the following properties.\footnote{To keep the notation consistent, we state Definition~\ref{def:srsp} for a given agent $i$. Strictly speaking, the subscript $i$ is not necessary.}

\begin{definition} \label{def:srsp}
Let $I$ be a nonempty interval in $\Re$.
A {\bf self-resampling procedure} 
with support $I$ and resampling probability $\mu\in(0,1)$
is a randomized algorithm with input
$b_i \in I$, random seed $\rM_i$, and output
	$\rY_i(b_i;\rM_i),\, \rZ_i(b_i;\rM_i) \in I$,
that satisfies the following properties:
\begin{enumerate}
\item \label{srsp-1}
For every fixed $\rM_i$,
	$\rY_i(b_i;\rM_i)$ and $\rZ_i(b_i;\rM_i)$
are non-decreasing functions of $b_i$.

\item \label{srsp-2}
With probability $1-\mu$, 
$\rY_i(b_i;\rM_i) = \rZ_i(b_i;\rM_i) = b_i$.
Otherwise 
$\rY_i(b_i;\rM_i) \leq \rZ_i(b_i;\rM_i) < b_i$.
\item \label{srsp-3}
The conditional
distribution of $\rY_i(b_i;\rM_i)$, given that
$\rZ_i(b_i;\rM_i) = b'_i < b_i$, is the same as the
unconditional distribution of $\rY_i(b'_i;\rM_i)$.
In other words,
\begin{align*}
	\Pr[\, \rY_i(b_i;\rM_i) < a_i  \;|\; \rZ_i(b_i;\rM_i) = b'_i \,]
		= \Pr[\, \rY(b'_i;\rM_i) < a_i \,], \quad
		\forall a_i \leq b'_i < b_i.
\end{align*}

\OMIT{ 
\mbedit{Can we give some intuition? Specifically, why we need both x and y.}
\BKnote{We can give some intuition, but not in the middle of the
formal definition.  I have added intuition on the first column of this page.
Do you think that the first column previously didn't give enough
intuition?  The sentence I added is confusing, and I would strongly
prefer to leave it out.  Hopefully someone else can come up with better
intuition.}} 

\item \label{srsp-4}
Consider the two-variable function
$$F(a_i,b_i) = \Pr[\rZ_i(b_i;\rM_i) < a_i \;|\; \rZ_i(b_i;\rM_i) < b_i],$$
which we will call the \emph{distribution function}
of the self-resampling procedure. For each $b_i$, the function $F(\cdot,b_i)$ must be differentiable and strictly increasing on the interval
$I \cap (-\infty,b_i)$.
\end{enumerate}
\end{definition}

As it happens, it is easier to construct self-resampling
procedures with support $\Re_+$, and one such construction
that we call the \emph{\CanonProc} (Algorithm~\ref{alg:canon-proc}) forms the basis for our general construction. We defer the discussion of self-resampling procedures with general support until after we have described and analyzed the generic transformation.

\newcommand{\ReplaceFunc}{\ensuremath{\mathtt{Recursive}}}
\newcommand{\WProb}[1]{{\bf with probability #1}}
\newcommand{\MyTab}{\hspace{4mm}}

\begin{center}
\begin{algorithm}[t]
\caption{The \CanonProc.}
\label{alg:canon-proc}
\begin{algorithmic}[1]
\STATE {\bf Input:} bid $b_i\in[0,\infty)$, parameter $\mu\in(0,1)$.
\STATE {\bf Output:} $(x_i,y_i)$ such that $0\leq x_i\leq y_i \leq b_i$.
\vspace{1mm}
\STATE \WProb{$1-\mu$}
\STATE \MyTab $x_i \leftarrow b_i$, $y_i \leftarrow b_i$.
\STATE {\bf else}
\STATE \MyTab Pick $b'_i \in [0,b_i]$ uniformly at random.
\STATE \MyTab $x_i \leftarrow \ReplaceFunc(b'_i)$, $y_i \leftarrow b'_i$.

\vspace{3mm}
\FUNC{$\ReplaceFunc(b_i)$}
\STATE \WProb{$1-\mu$}
\STATE \MyTab {\bf return} $b_i$.
\STATE {\bf else}
\STATE \MyTab Pick $b'_i \in [0,b_i]$ uniformly at random.
\STATE \MyTab {\bf return} $\ReplaceFunc(b'_i)$.
\ENDFUNC
\end{algorithmic}
\end{algorithm}
\end{center}

\begin{proposition} \label{prop:canon-proc}
Algorithm~\ref{alg:canon-proc} is a self-resampling procedure with support $\Re_+$ and resampling probability $\mu$. The distribution function for this  procedure is $F(a_i,b_i) = a_i/b_i$.
\end{proposition}

\begin{proof}
Properties~\ref{srsp-1}
and~\ref{srsp-2} in Definition~\ref{def:srsp} are
immediate from the description of the algorithm. The random seed $\rM_i$ for the algorithm can be defined as a countably infinite sequence of real numbers  drawn independently and uniformly at random from $[0,1]$ interval. Then in order pick a random number in some range $[0,r]$, the algorithm takes the next number in this sequence and multiplies it by $r$.

Property~\ref{srsp-3} follows from the recursive
nature of the sampling procedure:
the event
	$\rZ_i(b_i;\rM_i) = b'_i < b_i$
implies that the algorithm has followed the ``else" branch on Line 5, and has chosen $b'_i$ in Line 6. Finally, the distribution function is $F(a_i,b_i) = a_i/b_i$ since conditional on the event
$\rZ_i(b_i;\rM_i) < b_i$, the distribution of $\rZ_i(b_i;\rM_i)$
is uniform in the interval $[0,b_i]$.
Property~\ref{srsp-4} follows trivially.
\end{proof}

\subsection{The generic transformation}

Suppose we are given a monotone allocation rule
$\mA$ and for each agent $i \in N$ a self-resampling procedure that has
resampling probability $\mu \in (0,1)$, support $\types_i$, and output values $f_i=(\rY_i,\rZ_i)$.
Let $F_i(a_i,b_i)$ denote the distribution function of the
self-resampling procedure for agent $i$, and let
$F_i'(a_i,b_i)$ denote the partial derivative
$\frac{\partial F_i(a_i,b_i)}{\partial a_i}$.
Our generic transformation combines these ingredients
into a randomized mechanism
$\mathcal{M} = \GenericMech(\mA,\mu,\mathbf{f})$
that works as follows:

\renewcommand{\algorithmcfname}{Mechanism}
\begin{algorithm}[h]
\caption{Generic transformation
$\mathcal{M} = \GenericMech(\mA,\mu,\mathbf{f})$}
\label{mech:BKS-main}
\begin{algorithmic}[1]
\STATE Solicit bid vector $b \in \types$.
\STATE Execute each agent's self-resampling procedure
using an independent random seed $\rM_i$, to obtain two
vectors of modified bids
    \begin{align*}
    \mba &= (\rY_1(b_1;\rM_1) \LDOTS \rY_n(b_n;\rM_n)),\\
    \mbb &= (\rZ_1(b_1;\rM_1) \LDOTS \rZ_n(b_n;\rM_n)).
    \end{align*}
\STATE Allocate according to $\mA(\mba)$.
\STATE Each agent $i$ is charged
the amount $b_i \cdot \mA_i(\mba) - R_i$,
where $R_i$ is the \emph{rebate}
    \begin{align} \label{eq:rebate}
    R_i =
    \begin{cases}
        \frac{1}{\mu} \cdot \frac{\mA_i(\mba)}{F'_i(\mbb_i, b_i)}
           & \text{if $\mbb_i < b_i$}, \\
        0 & \text{otherwise}.
    \end{cases}
\end{align}
\end{algorithmic}
\end{algorithm}
\renewcommand{\algorithmcfname}{Algorithm}

If $\mA$ itself is randomized or if there is randomness arising from
nature, then we allocate according to $\mA(\mba;\rM,\rN)$
and we assume that the algorithm's random seed $\rM$ and the nature's random seed $\rN$ are independent of the random
seeds $\rM_i$ used in the resampling step.

We are now ready to present our main result:

\begin{theorem}\label{thm:maim-g}
Consider an arbitrary single-parameter domain.
Let $\mA$ be a monotone 
allocation rule.
Suppose we are given 
an ensemble $\mathbf{f}$ of self-resampling procedures $f_i = (\rY_i,\rZ_i)$
for each agent $i$, each with resampling probability $\mu \in (0,1)$.
Then the mechanism
	$\mathcal{M} = (\tA,\tP) = \GenericMech(\mA,\mu,\mathbf{f})$
has the following properties.
\begin{itemize}
\item[(a)] $\mathcal{M}$ is truthful, 
universally ex-post individually rational,
\item[(b)] For $n$ agents and any bid vector $b$ (and any fixed random seed of nature) allocations $\tA(b)$ and $\mA(b)$ are identical with probability at least $1-n\mu$.

\item[(c)] If $\types=\Re_+^n$ (all types are positive),
and each $f_i$ is the \CanonProc, then mechanism $\mathcal{M}$ is ex-post \noPositiveTransfers, and never pays any agent $i$ more than
	$b_i\cdot \mA_i(\mba)\cdot (\frac{1}{\mu}-1)$.
\OMIT{ 
Moreover, the modified types satisfy
	$\E[\rY_i(b_i;w_i)] = (1 - \frac{\mu}{2-\mu}) b_i$ for all $i$.
The welfare of $\mathcal{M}$ is at least
	$(1 - \frac{\mu}{2-\mu})$
times that of $\mA$.}
\end{itemize}
\end{theorem}

Several remarks are in order.

\begin{itemize}
\item The mechanism never explicitly computes the payment for each agent $i$ (\refeq{eq:myerson}) but rather implicitly creates the correct expected payments through its randomization of the bids.

\item The mechanism only invokes the original allocation rule $\mA$ {\em once}. This property is very useful when it is impossible to invoke the allocation rule more than once, e.g. for multi-armed bandit allocations.

\item The mechanism $\mathcal{M}$ is randomized even if $\mA$ is deterministic.
It is truthful in expectation over the randomness
used by the self-resampling procedures.

\item If $\mA$ is ex-post monotone, then $\mathcal{M}$ will be ex-post
truthful. To see this, fix nature's random seed $\rN$ and
apply Theorem~\ref{thm:maim-g} to the allocation rule $\mA_\rN$
induced by this $\rN$. 

\item If agents' types are positive then by part (b), the welfare of $\mathcal{M}$ is at least $1-n\mu$ times that of $\mA$. Further results on bounding the welfare loss are presented in Section~\ref{sec:welfare}.

\item By definition of the payment rule, the mechanism is universally ex-post normalized. We will not explicitly mention this property in the subsequent applications.
\end{itemize}

Parameter $\mu$ controls the trade-off between the loss in welfare and the variance in payments, as quantified by the rebate size $R_i$. If $\mu$ is very small and the mechanism issues rebate(s), then its revenue may be very low and possibly negative. However, this risk may be mitigated if the auction maker runs many independent auctions, as may be the case in practice. Further, the follow-up paper \cite{SingleCall-ec12} proves that our welfare vs. variance trade-off is optimal.

\OMIT{ 
Below we present the proof of parts (a) and (b) of Theorem~\ref{thm:maim-g}.
The proof of part (c) depends on the analysis of the \CanonProc{} in Section~\ref{sec:canon-proc-more}. We defer this proof to Section~\ref{sec:pf-positive-types}.
} 

\begin{proof}[of Theorem~\ref{thm:maim-g}]
We start with some notation. $\tA_i(b_{-i},b_i;\rT)$ denotes the allocation for agent $i$ given the bid vector $b = (b_{-i},b_i)$ and the combined random seed
$\rT = (\rM_1,\ldots,\rM_n,\rM,\rN)$. When we write $\tA_i(b_{-i},u)$ without
indicating the dependence on the $\rT$, we are referring to the unconditional expectation of $\tA_i(b_{-i},u;\rT)$ over $\rT$.

To prove that $\mathcal{M}$ is truthful, we need to prove
two things: that the randomized allocation rule $\tA$
is monotone, and that the expected payment rule $\tP$ satisfies
\begin{equation} \label{eq:tmyerson}
\tP_i(b) =
b_i \tA_i(b_{-i},b_i) -
\textstyle{\int_{-\infty}^{b_i}} \tA_i(b_{-i},u) \, du.
\end{equation}

\OMIT{ 
Recall that when we write $\tA_i(b_{-i},u)$ without
indicating the dependence on the combined random seed
$\rT = (\rM_1,\ldots,\rM_n,\rM,\rN)$, it means that we
are referring to the unconditional expectation
of $\tA_i(b_{-i},u;\rT)$.
}

The monotonicity of randomized allocation rule
$\tA$ follows from the monotonicity of $\mA$ and
the monotonicity property~\ref{srsp-1} in the definition
of a self-resampling procedure.
To prove that $\tP_i$ satisfies~\refeq{eq:tmyerson},
we begin by recalling that the payment charged to
player $i$ is $b_i \mA_i(\mba) - R_i$, where the rebate $R_i$
is defined by~\refeq{eq:rebate}.
The expectation of $b_i \mA_i(\mba)$
is simply $b_i \tA_i(b_{-i},b_i)$, so to conclude
the proof of truthfulness we must show that
\begin{equation} \label{eq:toshow-1}
\E[R_i] = \textstyle{\int_{-\infty}^{b_i}} \tA_i(b_{-i},u) \, du.
\end{equation}
Our proof of \refeq{eq:toshow-1} begins by
observing that the conditional distribution of
$\mba_i$, given that $\mbb_i=u<b_i$, is the same as
the unconditional distribution of $\rY_i(u;\rM_i),$
by Property 3 of a self-resampling
procedure.
Combining this with the fact that the random seed
$\rM_i$ is independent of $\{\rM_j : j \neq i\}$,
we find that the conditional distribution of the tuple
$\mba=(\mba_{-i},\mba_i)$, given that $\mbb_i=u$, is the
same as the unconditional distribution of the vector
$\hat{\mba}$ of modified bids that $\mathcal{M}$ would
input into the allocation rule $\mA$ if the bid vector
were $(b_{-i},u)$ instead of $(b_{-i},b_i)$.  Taking
expectations, this implies that for all $u < b_i,$ we have
$\E[\mA_i(\mba) \,|\, \mbb_i=u] = \E[\mA_i(\hat{\mba})] =
\tA(b_{-i},u).$

Now apply Theorem~\ref{thm:estimator} with the
function $g(u) = \tA_i(b_{-i},u).$  Recalling that
$F_i(\cdot,b_i)$ is the cumulative distribution
function of $\mbb_i$ given that $\mbb_i < b_i$,
we apply the theorem to obtain
\begin{align} \label{eq:maim-g-2}
\int_{-\infty}^{b_i} \tA_i(b_{-i},u) \, du
&=
\E \left[ \left.
\frac{\tA_i(b_{-i},\mbb_i)}{F'_i(\mbb_i,b_i)}
\,\right|\,
\mbb_i < b_i \right] \nonumber
=
\E \left[ \left.
\frac{\mA_i(\mba)}{F'_i(\mbb_i,b_i)} \,\right|\, \mbb_i < b_i \right] \nonumber \\
&=
\mu \cdot \E[R_i \,|\, \mbb_i < b_i],
\end{align}
where the second equation follows from the equation
derived at the end of  the preceding paragraph, averaging over all $u < b_i.$
Observing that $R_i = 0$ unless $\mbb_i < b_i$, an event
that has probability $\mu$, we see that
$\E[R_i] = \mu \cdot \E[R_i \,|\, \mbb_i < b_i].$
Combined with \refeq{eq:maim-g-2}, this establishes
\refeq{eq:toshow-1} and completes the proof that
$\mathcal{M}$ is truthful.

Mechanism $\mathcal{M}$ is universally
ex-post individually rational because agent $i$
is never charged an amount greater than $b_i \tA_i(b;q)$.
Part (b) follows from the union bound: the probability
that $\mba_i = b_i$ for all $i$ is at least $1-n\mu$.
For part (c), note that by Proposition~\ref{prop:canon-proc},
the canonical self-resampling procedure has distribution function $F(a_i,b_i) = a_i/b_i$, hence
$F'_i(\mbb_i,b_i) = 1/b_i,$ for all $i, \mbb_i, b_i.$  The
rebate $R_i$ is  equal either to $0$ or to
	$\frac{1}{\mu} \cdot \frac{A_i(\mba)}{F'_i(\mbb_i,b_i)}
		= b_i \cdot A_i(\mba) \cdot \frac{1}{\mu}.$
We also charge $b_i \cdot A_i(\mba)$ to
agent $i$.  The claimed upper bound on the amount paid to
agent $i$ follows by combining these two terms.
\end{proof}

\subsection{Self-resampling procedures with general support}
\label{sec:canon-proc-more}

To construct a self-resampling procedure with
support in an arbitrary interval $I$, we can use the
following technique.  Suppose $h: (0,1] \times I \rightarrow I$
is a two-variable function such that the partial derivatives
$\partial h(z_i,b_i) / \partial z_i$ and
$\partial h(z_i,b_i) / \partial b_i$ are
well-defined and strictly positive at
every point $(z_i,b_i)\in (0,1] \times I$.  Suppose furthermore that
$h(1,b_i) = b_i$ and $\inf_{z_i\in (0,1]} \{h(z_i,b_i)\} = \inf(I)$ for all $b_i \in I$.
Then we define the \emph{$h$-\CanonProc}
	$(\rY^h_i, \rZ^h_i)$
with support $I$, by specifying that
\begin{align}\label{eq:h-canon-proc}
\left\{ \begin{array}{rcl}
\rY^h_i(b_i;\rM_i) 	&= h(\rY_i(1;\rM_i),b_i) \\
\rZ^h_i(b_i;\rM_i) 	&= h(\rZ_i(1;\rM_i),b_i),
\end{array} \right.
\end{align}
where $(\rY_i,\rZ_i)$ is the \CanonProc{} as defined in Algorithm~\ref{alg:canon-proc}.

\begin{proposition} \label{prop:h-canon-proc}
$(\rY^h_i, \rZ^h_i)$ as defined in~\refeq{eq:h-canon-proc} is a self-resampling procedure with support $I$ and resampling probability $\mu$. The distribution function for $(\rY^h_i, \rZ^h_i)$ is the unique two-variable function $F(a_i,b_i)$ such that
\begin{align}\label{eq:prop-h-canon-proc}
	h(F(a_i,b_i), b_i) = a_i
		\quad\text{for all $a_i,b_i \in I, \,a_i < b_i$}.
\end{align}
\end{proposition}

\begin{proof}
Property~\ref{srsp-1} in Definition~\ref{def:srsp} holds because of the
monotonicity of $h$, Property~\ref{srsp-2}
holds because $h(1,b_i)=b_i$ for all $b_i$,
and Property~\ref{srsp-3} holds because
the function $h$ is deterministic
and monotone.

Let $F_h(a_i,b_i)$ and $F_0(a_i,b_i)$ be the distribution functions for the $h$-canonical and canonical self-resampling procedures, respectively. Recall that 	
	$F_0(a_i,b_i) = a_i/b_i $
by Proposition~\ref{prop:canon-proc}. Note that $F(a_i,b_i)$ in~\refeq{eq:prop-h-canon-proc} is unique (and hence well-defined) by the strict monotonicity of $h$.

The claim that
	$F_h(a_i,b_i) = F(a_i,b_i)$
easily follows from in~\refeq{eq:h-canon-proc}. By definition of $h$ we have
\begin{align*}
	h(y_i(1,w_i),\, b_i) < b_i \iff y_i(1,w_i) <1.
\end{align*}
Therefore, letting
	$y_i = y_i(1,w_i)$
we have
\begin{align*}
F_h(a_i,b_i)
	&\triangleq \Pr[ h(y_i,b_i) <a_i \,|\, h(y_i,b_i) < b_i] \\
	& = \Pr[ y_i < F(a_i,b_i) \,|\, y_i <1] \\
	& = F_0(\, F(a_i,b_i)\, ,\, 1) \\
	&= F(a_i,b_i).
\end{align*}

\noindent Our assumption that $h$ is differentiable and strictly
increasing in its first argument now implies that
the same property holds for $F$, which verifies
Property~\ref{srsp-4}.
\end{proof}

\subsection{A simplified generic transformation for positive types}
\label{sec:simpler-transformation}


We focus on the important special case of positive types, and present Mechanism~\ref{mech:BKS-simple}, a simplified version of the generic transformation (Mechanism~\ref{mech:BKS-main}), for this case.

\renewcommand{\algorithmcfname}{Mechanism}
\begin{algorithm}[h]
\caption{A simplified generic transformation for positive types.}
\label{mech:BKS-simple}
\begin{algorithmic}[1]
\STATE {\bf Parameter:} resampling probability $\mu\in(0,1)$. \vspace{2mm}
\STATE Collect bid vector $b\in (0,\infty)^n$.
\STATE Independently for each agent $i \in [n]$:
\STATE\TAB Sample: $\gamma_i$ uniformly at random from $[0,1]$
\STATE\TAB Set $\chi_i=1$ with probability $1-\mu$ and otherwise $\chi_i=\gamma_i^{1/(1-\mu)}$.
\STATE Construct the vector of modified bids $ x = (x_1, \ldots, x_n)$, where
    $x_i = \chi_i\, b_i$.
\STATE Allocate according to $\mA(x)$.
\STATE For each agent $i$, assign payment
    $b_i \cdot \mA_i(x) \cdot \begin{cases}
    1 & \mbox{if $\chi_i = 1$}, \\
    1 - \frac{1}{\mu} & \mbox{if $\chi_i < 1$}
    \end{cases}.$
\end{algorithmic}
\end{algorithm}
\renewcommand{\algorithmcfname}{Algorithm}

We prove that Mechanism~\ref{mech:BKS-simple} is equivalent to the generic transformation (Mechanism~\ref{mech:BKS-main}) with a \CanonProc (Algorithm~\ref{alg:canon-proc}).

\begin{proposition}\label{prop:simplified-mech}
The allocation and payments in Mechanism~\ref{mech:BKS-simple} coincide with those in
Mechanism~\ref{mech:BKS-main} with a \CanonProc.
\end{proposition}

To prove Proposition~\ref{prop:simplified-mech}, we provide a non-recursive version of the \CanonProc (Algorithm~\ref{alg:canon-proc}), which we call $\SecondAlg$. We argue that the output of  Mechanism~\ref{mech:BKS-simple} is identical to the output of the mechanism obtained by plugging $\SecondAlg$ into Mechanism~\ref{mech:BKS-main}.%
\footnote{This observation is due to~\cite{Parkes-netecon12}.}
$\SecondAlg$ is also essential for the analysis in Section~\ref{sec:welfare}.

\begin{algorithm}[h]
\caption{$\SecondAlg$: a non-recursive version of $\FirstAlg$.}
\label{alg:canon-nonrecursive}
\begin{algorithmic}[1]
\STATE {\bf Input:} bid $b_i\in[0,\infty)$, parameter $\mu\in(0,1)$.
\STATE {\bf Output:} $(x_i,y_i)$ such that $0\leq x_i\leq y_i \leq b_i$.
\vspace{1mm}
\STATE \WProb{$1-\mu$}
\STATE \MyTab $x_i \leftarrow b_i$, $y_i \leftarrow b_i$.
\STATE {\bf else}
\STATE \MyTab Pick $\gamma_1, \gamma_2 \in [0,1]$ indep., uniformly at random.
\STATE \MyTab $x_i \leftarrow b_i\cdot \gamma_1^{1/(1-\mu)}$,~~
		 $y_i \leftarrow b_i\cdot \max\{ \gamma_1^{1/(1-\mu)}, \gamma_2^{1/\mu} \}$.
\end{algorithmic}
\end{algorithm}

\begin{proposition}
$\FirstAlg$ and $\SecondAlg$ generate the same output distribution:
 for any bid $b_i\in [0,\infty)$, the joint distribution of the pair
	$$(\rY_i,\rZ_i) = (\rY_i(b_i;\rM_i),\rZ_i(b_i;\rM_i))$$
is the same for both procedures.
(Here $\rM_i$ denotes the random seed for each agent $i$.)
\label{prop:srsp}
\end{proposition}

The proof of Proposition~\ref{prop:srsp} can be found in the Appendix.

\begin{proof}[of Proposition~\ref{prop:simplified-mech}]
By Proposition~\ref{prop:srsp}, it suffices to compare Mechanism~\ref{mech:BKS-simple} to Mechanism~\ref{mech:BKS-main} with self-resampling procedure $\SecondAlg$.
To show that the two mechanisms are equivalent, we must show that they
yield the same distribution over allocations and the same payments.
First we argue about the allocations. In both mechanisms, each bidder's
bid $b_i$ is independently transformed into a random $x_i$, and then the
allocation rule $\mA$ is applied to the vector $x=(x_1,\ldots,x_n)$.
Furthermore, the conditional distribution of $x_i$ given $b_i$ is the
same in both cases: $x_i=b_i$ with probability $1-\mu$, and otherwise
$x_i = b_i \cdot \gamma^{1/(1-\mu)}$ where $\gamma$ is uniformly distributed
in $[0,1]$. Hence, the two mechanisms yield the same distribution
over allocations.

To see that the payment rules are the same,
consider the distribution function of $\SecondAlg$,
as defined in Definition~\ref{def:srsp}:
\begin{align*}
F(a_i,b_i)  = \Pr[\rZ_i(b_i;\rM_i) < a_i \;|\; \rZ_i(b_i;\rM_i) < b_i].
\end{align*}
By Proposition~\ref{prop:srsp}, $F_i(a_i, b_i)$ is also the distribution function for $\FirstAlg$. By Proposition~\ref{prop:canon-proc} we have
$F(a_i,b_i) = a_i/b_i$, and consequently
\begin{align*}
    F_i'(a_i,b_i) \triangleq \frac{\partial F_i(a_i,b_i)}{\partial a_i} =
\frac{1}{b_i}.
\end{align*}
In particular, neither allocation nor payments in this mechanism depend on the $y_i$'s. Suppressing the $y_i$'s from mechanism $\mathcal{M}$ and plugging in $F'_i(\mbb_i, b_i)= \frac{1}{b_i}$, we obtain Mechanism~\ref{mech:BKS-simple}. This completes the proof of Proposition~\ref{prop:simplified-mech}.
\end{proof}

\section{Improved bounds on welfare}
\label{sec:welfare}


We present improved bounds on the welfare obtained by our generic transformation. We consider two interesting special cases when the agents' private types are, respectively, always positive and always negative. In the second case, agents are contractors who incur costs and get paid by the mechanism; one such example is a shortest paths mechanism considered in Section~\ref{sec:apps}.

We consider the approximation that is achieved by the mechanism as a function of the approximation of the original allocation rule. Recall that our generic transformation creates a mechanism with an allocation that is identical to the original allocation with probability at least $1-n\mu$. For positive types this immediately implies a bound on the approximation which degrades with $n$, the number of agents. (For negative types such bound does not immediately follow since the cost in the low probability event might be prohibitively high.) For both settings, we present a similar bound that does \emph{not} degrade with $n$.

%
%
%

\subsection{Positive private types}

Assume that the agents' types are always positive, more specifically that the type space is
    $\types = (0,\infty)^n$.
Recall that for agents' types $t\in \types$ the social welfare of an outcome $o$ is defined to
	$\SW(o, t) = \sum_{i\in N} t_i\, a_i(o)$.
The optimal social welfare is
	$\OPT(t) = \max_{o\in O} \SW(o, t)$,
where $O$ is the set of all feasible outcomes.
(A mechanism with) an allocation rule $\mA$ is {\em $\alpha$-approximate}
if it holds that
\begin{align}\label{eq:def-approx}
	\alpha\cdot \E[\SW(\mA(t), t)] \ge \OPT(t)
		\text{ for every $t$}.
\end{align}

\begin{theorem} \label{thm:improved-approximation}
Consider the setting in Theorem~\ref{thm:maim-g}(c), so that
    $\types = (0,\infty)^n$ and each $f_i$ is the \CanonProc.
If allocation rule $\mA$ is $\alpha$-approximate,
then mechanism $\GenericMech(\mA,\mu,\mathbf{f})$
is $\alpha/(1-\frac{\mu}{2-\mu})$-approximate.
\end{theorem}


\begin{proof}
Fix a bid vector $b$, and let $o^*$ be the corresponding optimal allocation. Recall that our mechanism outputs allocation $\mA(\mba)$, where $x$ is the vector of randomly modified bids.
As the original allocation rule $\mA$ is $\alpha$-approximate,
by~\refeq{eq:def-approx} it holds that
$\alpha \cdot \SW(\mA(\mba), \mba) \ge \OPT(\mba)$.
We will show that
\begin{equation}\label{eq:shrink}
\E[\mba_i] =  \left(1-\tfrac{\mu}{2-\mu}\right) b_i
	\text{~~for each agent $i$}.
\end{equation}
Thus when we evaluate $o^*$ with respect to bids $\mba$ we get:
\begin{align*}
\alpha \cdot \SW(A(\mba), \mba)
	&\ge \OPT(\mba)
	 \ge \SW(o^*,\mba)
	 = \textstyle{\sum_{i\in N}}\, \mba_i\, a_i(o^*) \\
\alpha \cdot \E[\SW(A(\mba), \mba)]
    &= \E\left[ \textstyle{\sum_{i\in N}}\, \mba_i\, a_i(o^*)  \right] \\
	&= \sum_{i\in N} \left(1-\frac{\mu}{2-\mu}\right) b_i\, a_i(o^*)
	= \left(1-\frac{\mu}{2-\mu}\right)\OPT(b).
\end{align*}

It remains to prove~\refeq{eq:shrink}. Let us use $\SecondAlg$ to describe the \CanonProc. Recall that $\SecondAlg$ generates
$\mba_i = \rY_i(b_i;\rM_i)$ by setting $\mba_i = b_i$
with probability $1-\mu$, and otherwise sampling $\gamma_1$
uniformly at random in $[0,1]$ and outputting
$\mba_i = b_i \cdot \gamma_1^{1/(1-\mu)}.$  Hence
\begin{align*}
\E[\mba_i \,|\, \mba_i < b_i]
&= \textstyle{\int_{0}^{1}}\, b_i \cdot \gamma_1^{1/(1-\mu)} \, d\gamma_1
= b_i \cdot \frac{1}{1+\frac{1}{1-\mu}}
= b_i \cdot \left( 1 - \tfrac{1}{2-\mu} \right) \\
\E[\mba_i]
	&= (1-\mu)\cdot b_i + \mu\cdot \E[\mba_i \,|\, \mba_i < b_i]
= b_i \cdot \left( 1 - \tfrac{\mu}{2-\mu} \right). \qquad \qedhere
\end{align*}
\end{proof}

For arbitrary self-resampling procedures $f_i$ with support $\Re_+$,
~\refeq{eq:shrink} can be replaced by
	$\E[\mba_i] \ge  \left(1-\mu\right) b_i$,
which gives a slightly weaker result, namely an $\tfrac{\alpha}{1-\mu}$-approximation to the social welfare.


\subsection{Negative private types}
\label{sec:negative-types}

Now assume that the agents' types are always negative, more specifically that $\types = (-\infty,0)^n$.
For negative types approximation is defined with respect to the social cost, which is the negation of the social welfare. An algorithm is  \emph{$\alpha$-approximate} if for every input it outputs an outcome with cost at most $\alpha$ times the optimal cost. We present an approximation bound for an $h$-\CanonProc{}, for a suitably chosen $h$.

\begin{theorem} \label{thm:approximation-negative}
Consider the setting in Theorem~\ref{thm:maim-g}. Assume that $\types = (-\infty,0)^n$ and that each $f_i$ is the $h$-\CanonProc, where
	$h(z_i,b_i) = b_i/\sqrt{z_i}$.
Suppose $\mu\in (0,\tfrac12)$.
If allocation rule $\mA$ is $\alpha$-approximate,
then mechanism $\GenericMech(\mA,\mu,\mathbf{f})$
is $\alpha\left( 1 + \tfrac{\mu}{1-2\mu} \right)$-approximate.
\end{theorem}

The proof of this theorem is almost identical to that of Theorem~\ref{thm:improved-approximation}, and thus is omitted. The main  modification is that~\refeq{eq:shrink} is replaced by the following lemma:

\begin{lemma}
\label{lem:expected-blow}
In the setting of Theorem~\ref{thm:approximation-negative}, letting $x^h$ be the vector of modified types, it holds that
\begin{align*}
	\E[\mba^h_i] = b_i \left(1+\tfrac{\mu}{1-2\mu} \right)
	\text{~~for all $i$}.
\end{align*}
\end{lemma}

\begin{proof}
Recall that $x^h$ is defined by~\refeq{eq:h-canon-proc}.
As in the proof of~\refeq{eq:shrink}, we will use $\SecondAlg$ to describe the \CanonProc.
It follows that
\begin{align*}
\E[\mba^h_i \,|\, \mba_i < b_i]
&= \textstyle{\int_{0}^{1}}\, \frac{b_i}{\sqrt{\gamma_1^{\frac{1}{(1-\mu)}}}} \, d\gamma_1
= \textstyle{\int_{0}^{1}}\, b_i \cdot \gamma_1^{-\frac{1}{2(1-\mu)}} \, d\gamma_1
= b_i \cdot \frac{1}{1-\frac{1}{2(1-\mu)}}
= b_i \cdot\left( 1 + \tfrac{1}{1-2\mu} \right),\\
\E[\mba^h_i]
	&= (1-\mu) \cdot b_i + \mu\cdot  \E[\mba^h_i \,|\, \mba^h_i < b_i]
= b_i \cdot \left( 1 + \tfrac{\mu}{1-2\mu} \right). \qquad\qedhere
\end{align*}
\end{proof}



\OMIT{ 
We begin by constructing a pair of canonical self-resampling procedures
with support $(0,\infty)$ and then we use this as an ingredient in
constructing self-resampling procedures with arbitrary support.
Assume we are given a fixed parameter $\mu\in (0,1)$.  The canonical
self-resampling procedures will use a random seed $\rM$
consisting of two infinite sequences $(\beta_1,\beta_2,\ldots)$
and $(\gamma_1,\gamma_2,\ldots)$ such that the
$\beta_j$ are i.i.d.~Bernoulli random variables with
$\E(\beta_j) = 1-\mu,$ and the $\gamma_j$ are i.i.d.~uniform
random variables in $[0,1].$  We  have the following
two procedures.
\begin{description}
\item[Algorithm $\FirstAlg(b;\rM)$:]
If $\beta_1 = 1$, output $\rY(b;\rM) = \rZ(b;\rM) = b$.
Otherwise, let $b' = b \cdot \gamma_1$, and let
$\rM'$ denote the \emph{shifted random seed}
consisting of the sequences $(\beta_2,\beta_3,\ldots)$
and $(\gamma_2,\gamma_3,\ldots).$  Output
$\rY(b;\rM) = \rY(b';\rM')$ and
$\rZ(b;\rM) = b'.$  Note that the recursive
procedure for calculating $\rY$ terminates as
long as $\beta_j=1$
for at least one value of $j$,
an event that has probability $1$.
\item[Algorithm $\SecondAlg(b;\rM)$:]
If $\beta_1=1$, output $\rY(b;\rM) = \rZ(b;\rM) = b$.
Otherwise, output $\rY(b;\rM) = b \cdot \gamma_1^{1/(1-\mu)}$
and $\rZ(b;\rM) = b \cdot \max\{\gamma_1^{1/(1-\mu)},\gamma_2^{1/\mu}\}.$
\end{description}

\begin{proposition} \label{prop:srsp}
Both $\FirstAlg$ and $\SecondAlg$ are self-resampling
procedures with resampling probability $\mu$.  In fact, they
generate the same output distribution.  In other words,
for all $b$, the joint distribution of the pair $(\rY,\rZ) =
(\rY(b;\rM),\rZ(b;\rM))$ is the same regardless of
whether one computes $(\rY,\rZ)$ using
$\FirstAlg$ or $\SecondAlg$.
The distribution function for both procedures
is $F(u,v) = u/v.$
\end{proposition}

The proof of Proposition~\ref{prop:srsp} is in
Appendix~\ref{app:proof-prop:srsp}. 
For now, we
observe that it is easy to show that $\FirstAlg$
is a self-resampling procedure.  Properties~\ref{srsp-1}
and~\ref{srsp-2} in Definition~\ref{def:srsp} are
immediate from the description of the algorithm.
Property~\ref{srsp-3} follows from the recursive
nature of the sampling procedure:
the event $\rZ(b;\rM) = b' < b$ implies
that $\beta_1 = 0$ and $\gamma_1 = b'$.  In
that case, $\rY(b;\rM) = \rY(b';\rM')$
and the conditional distribution of $\rY(b';\rM')$ is the same
as the unconditional distribution of
$\rY(b';\rM)$ because $\rM'$ and $\rM$
have the same distribution and
$\rM'$ is independent of the event $\rZ(b;\rM) = b'$.
Finally, $\FirstAlg$ has distribution function $F(u,v) = u/v$ since conditional on the event
$\rZ(b;\rM) < b$, the distribution of $\rZ(b;\rM)$
is uniform in the interval $[0,b]$.
Property~\ref{srsp-4} is now immediate
because $F(u,v)=u/v$ is a differentiable,
strictly increasing function of $u$.

Now, to construct a self-resampling procedure with
support in an arbitrary interval $I$, we can use the
following technique.  Suppose $h: (0,1] \times I \rightarrow I$
is a two-variable function such that
$\partial h(z,b) / \partial z$ and
$\partial h(z,b) / \partial b$ are
well-defined and strictly positive at
every point $(z,b)$ in $(0,1] \times I$, the domain of
definition of $h$.  Suppose furthermore that
$h(1,b) = b$ and $\inf_{z\in (0,1]} \{h(z,b)\} = \inf(I)$ for all $b \in I$.
Then we define a self-resampling procedure
$(\rY^h_\mu, \rZ^h_\mu)$ with support $I$, by specifying
that $\rY^h_\mu(b;\rM) = h(\rY(1;\rM),b)$
and $\rZ^h_\mu(b;\rM) = h(\rZ(1;\rM),b)$,
where $\rY$ and $\rZ$ are the canonical
self-resampling procedure computed by
either $\FirstAlg$ or $\SecondAlg$.
Property~\ref{srsp-1} holds because of the
monotonicity of $h$, Property~\ref{srsp-2}
holds because $h(1,b)=b$ for all $b$,
and Property~\ref{srsp-3} holds because
the function $h$ is deterministic
and monotone.

The distribution function of this self-resampling procedure
$(\rY^h_\mu, \rZ^h_\mu)$ can easily be determined, since
Proposition~\ref{prop:srsp} gives the distribution function
for the canonical self-resampling procedure, and
$(\rY^h_\mu, \rZ^h_\mu)$ is obtained from it by a
simple change of variables.  In fact, the distribution
function $F(u,v)$ of $(\rY^h_\mu, \rZ^h_\mu)$ is the
unique two-variable function
such that $h(F(u,v), v) = u$ for all $u,v \in I, \,
u < v$.
Our assumption that $h$ is differentiable and strictly
increasing in its first argument now implies that
the same property holds for $F$, which verifies
Property~\ref{srsp-4} of a self-resampling procedure.
} 

\section{Applications to offline mechanism design}
\label{sec:apps}

\xhdr{The VCG mechanism for shortest paths.}
The seminal paper \citeN{NR01} has presented the following question: is there a computational overhead in computing payments that will induce agents to be truthful, compared to the computation burden of computing the allocation. One of their examples is the VCG mechanism for the \emph{shortest path mechanism design problem}, where a naive computation of VCG payments requires additional computation of $n$ shortest path instances. Yet, an explicit payment computation is not the real goal, it is just a means to an end. The real goal is {\em inducing the right incentives}. Our procedure shows that without {\em any} overhead in computation, if we move to a randomized allocation rule and settle for truthfulness in expectation (and a small loss in performance) one can induce the right incentives.

The shortest path mechanism design problem is the following. We are given a graph $G=(V,E)$ and a pair of source-target nodes $(v_s,v_t)$. Each agent $e$ controls an edge $e\in E$ and has a cost $c_e> 0$ if picked (thus $v_e=-c_e<0$ and $\types_e=(-\infty,0)$ for every $e$). That cost is private information, known only to agent $e$.  The mechanism designer's goal is to pick a path $P$ from node $v_s$ to node $v_t$ in the graph with minimal total cost, that is $\sum_{e\in P} c_e$ is minimal. Assume that there is no edge that forms a cut between
$v_s$ and $v_t$.

The VCG mechanism is an cost-optimal and truthful mechanism for this problem. It computes a shortest path $P$ with respect to the reported costs and pays to an agent $e$ the difference between the cost of the shortest path that does not contains $e$ and the total cost shortest path excluding the cost of $e$. A naive implementation of the VCG mechanism requires computing $|P|+1$ shortest path instances (where $|P|$ denotes the number of edges in path $P$). VCG is deterministic, truthful and cost-optimal.

\newcommand{\EFF}{\ensuremath{\mathtt{EFF}}}

Let \EFF{} an the cost-optimal allocation rule for the shortest path problem. We can use our general procedure to derive the following result (its proof follows directly from Theorem~\ref{thm:maim-g} and Theorem~\ref{thm:approximation-negative}).

\begin{theorem}\label{thm:shortest-paths}
Fix any $\mu\in (0,\half)$.
For each agent $i$, let $f_i$ be the $h$-\CanonProc, where
	$h(z_i,b_i) = b_i/\sqrt{z_i}$.
Let
	$\mathcal{M} = \GenericMech(\EFF,\mu,\{f_i\})$
be the mechanism created by applying $\GenericMech()$ to \EFF. Then $\mathcal{M}$ has the following properties:
\begin{OneLiners}
\item It is truthful and universally individually rational.

\item It only computes one shortest paths instance.
\item It outputs a path with expected length at most
$\left( 1 + \tfrac{\mu}{1-2\mu} \right)$
times the length of the shortest path.

    \OMIT{Thus for any $\epsilon>0$ we can get $1+\epsilon$-approximation by picking $\mu=\frac{\epsilon}{1+2\epsilon}$.}


\OMIT{ 
\item With probability $\mu$ an agent gets paid more than his cost. When this happens, for $y_i\in [0,1]$ picked uniformly at random the agent is being paid $\frac{1}{2\mu y^{3/2}_i}$ times his value of the allocation. Thus for any fixed $y_i$ the payment grows inversely proportional to $\mu$.
} 
\end{OneLiners}
\end{theorem}

Recall that parameter $\mu$ controls the trade-off between approximation ratio and the rebate size $R_i$, which for a given random seed is proportional to $\tfrac{1}{\mu}$.

\xhdr{Communication overhead of payment computation.}
\citeN{BabaioffBS13} show that there exists a monotone deterministic allocation rule for which
the communication required for computing the allocation is factor $\Omega(n)$ less than the communication required to computing prices. This implies that inducing the correct incentives {\em deterministically} has a large overhead in communication.
Assume that instead of requiring explicit computation of payments we are satisfied with inducing the correct incentives using a {\em randomized} mechanism.
In such case our reduction shows that the deterministic lower bound cannot be extended to randomized mechanisms, if we allow a small error in the allocation.

More concretely, consider a single parameter domain with types that are positive, $\types_i=(0,\infty)$ (as in~\cite{BabaioffBS13}).  For all $i$, use the \CanonProc.
Consider any monotone allocation rule $\mA$. We can apply Theorem~\ref{thm:maim-g} to obtain a randomized mechanism that is truthful and only executes that allocation rule $\mA$ {\em once} (thus has no communication overhead at all) and has exactly the same allocation with probability at least $(1-\mu)^n$. For any $\epsilon>0$ we can find $\mu>0$ such that the error probability is less than $\epsilon$.

\section{Multi-armed bandit mechanisms}
\label{sec:MAB}

\newcommand{\payoff}{click reward\xspace}
\newcommand{\payoffs}{click rewards\xspace}
\newcommand{\realization}{click realization\xspace}
\newcommand{\realizations}{click realizations\xspace}
\newcommand{\Realizations}{Click realizations\xspace}

\newtheorem{claim}[theorem]{Claim}

\newcommand{\RefMainThm}{Theorem~\ref{thm:maim-g}(c)}
\newcommand{\Active}{S_{\mathtt{act}}}

In this section we apply the main result to \emph{multi-armed bandit (MAB) mechanisms}: single-parameter mechanisms in which the allocation rule is (essentially) an MAB algorithm parameterized by the bids. As in any single-parameter mechanism, agents submit their bids, then the allocation rule is run, and then the payments are assigned. This application showcases the full power of the main result, since in the MAB setting the allocation rule is only run once, and (in general) cannot be simulated as a computational routine without actually implementing the allocation.

Focusing on the stochastic setting, we design truthful MAB mechanisms with the same regret guarantees as the best MAB \emph{algorithms} such as \UCB~\cite{bandits-ucb1}. First, we prove that allocation rules derived from \UCB{} and similar MAB algorithms are in fact monotone, and hence give rise to truthful MAB mechanisms. Second, we provide a new allocation rule with the same regret guarantees that is ex-post monotone, and hence gives rise to an \emph{ex-post} truthful MAB mechanism. Third, we use this new allocation rule to obtain an unconditional separation between the power of randomized and deterministic ex-post truthful MAB mechanisms.

\OMIT{ 
The section is structured as follows. In Section~\ref{sec:MAB-prelims} we overview the relevant background on multi-armed bandits. Then in Section~\ref{sec:MAB-monotonicity} we show that a very general class of MAB algorithms gives rise to monotone MAB allocations (to which therefore \RefMainThm{} can be applied). In Section~\ref{sec:MAB-stochastic} we consider the stochastic MAB setting; we present the background story,  apply the general result from Section~\ref{sec:MAB-monotonicity} to an algorithm from the literature, and present a new algorithm for a further improvement. Finally, in Section~\ref{sec:separation} } 

\subsection{Preliminaries: MAB mechanisms}
\label{sec:MAB-prelims}

An MAB mechanism~\cite{MechMAB-ec09,DevanurK09} operates as follows. There are $n$ agents. Each agent $i$ has a private value $v_i$ and submits a bid $b_i$. We assume that $b_i,v_i \in [0, \bmax]$, where $\bmax$ is known a priori. The allocation consists of $T$ rounds, where $T$ is the \emph{time horizon}. In each round $t$ the allocation rule chooses one of the agents, call it $i=i(t)$, and observes a \emph{\payoff} $\pi(t) \in [0,1]$ for this choice; the chosen agent $i$ receives $v_i\, \pi(t)$ units of utility. Payments are assigned after the last round of the allocation. Note that the social welfare of the mechanism is equal to the total value-adjusted \payoff:
$\textstyle{\sum_{t=1}^T}\, v_{i(t)}\, \pi(t)$.

The special case of 0-1 \payoffs{} corresponds to the scenario in which agents are advertisers in a pay-per-click auction, and choosing agent $i$ in a given round $t$ means showing this agent's ad. Then the \payoff{} $\pi(t)$ is the \emph{click bit}: $1$ if the ad has been clicked, and $0$ otherwise. Following the web advertising terminology, we will say that in each round, an \emph{impression} is allocated to one of the agents.

\newcommand{\ALG}{\ensuremath{\hat{\mA}}}

Formally, an \emph{MAB allocation rule} $\mA$ is an online algorithm parameterized by $n,T,\bmax$ and the bids $b$. In each round it allocates the impression and observes the \payoff. Absent truthfulness constraints, the objective is to maximize the \emph{reported welfare}:
	$\textstyle{\sum_{t=1}^T}\, b_{i(t)}\, \pi(t)$.
This formulation generalizes  MAB \emph{algorithms}: the latter are precisely MAB allocation rules with all bids set to $1$.

Given an MAB algorithm $\ALG$, there is a natural way to transform it into an MAB allocation rule $\mA$. Namely, $\mA$ runs algorithm $\ALG$ with modified \payoffs: if agent $i$ is chosen in round $t$ then the \payoff{} reported to $\ALG$ is
	$\hat{\pi}(t) = (b_i/ \bmax)\, \pi(t) $.
We will say that algorithm $\ALG$ \emph{induces} allocation rule $\mA$. From now on we will identify an MAB algorithm with the induced allocation rule, e.g. allocation rule \UCB{} is induced by algorithm \UCB~\cite{bandits-ucb1}.

We will focus on the \emph{stochastic} MAB setting: in all rounds $t$ in which an agent $i$ is chosen, the \payoff{} $\pi(t)$ is an independent random sample from some fixed distribution on $[0,1]$ with expectation $\mu_i$.\footnote{The exact shape of this distribution is not essential. E.g. in the advertising example $\pi(t)\in \{0,1\}$.}
Following the web advertisement terminology,
we will call $\mu_i$ the \emph{click-through rate} (CTR) of agent $i$. The CTRs are fixed, but no further information about them (such as priors) is revealed to the mechanism.

\xhdr{Regret.}
The performance of an MAB allocation rule is quantified in terms of \emph{regret}:
\begin{align*}
R(T;b;\mu) \triangleq T\, \textstyle{\max_i} [\, b_i\,\mu_i \,] -
	\E[\, \textstyle{\sum_{t=1}^T}\; b_{i(t)}\, \mu_{i(t)} \,],
\end{align*}
the difference in expected \payoffs{} between the algorithm and the benchmark: the best agent in hindsight, knowing the $\mu_i$'s. We focus on
	$R(T) \triangleq \max R(T;b;\mu)$,
where the maximum is taken over all CTR vectors $\mu$ and all bid vectors $b$ such that $b_i\leq 1$ for all $i$.~\footnote{\label{ft:Bmax}We define  $R(T)$ with $\bmax= 1$ merely to simplify the notation. All regret bounds (scaled up by a factor of $\bmax$) hold for an arbitrary $\bmax$.}

Regret guarantees from the vast literature on MAB algorithms easily translate to MAB allocation rules. In particular, allocation rule \UCB{} has regret
	$R(T) = O(\sqrt{nT \log T})$~\cite{bandits-ucb1},
which is nearly matching the information-theoretically optimal regret bound $\Theta(\sqrt{nT})$~\cite{bandits-exp3,Bubeck-colt09}.
The stochastic MAB setting tends to be easier if the best agent is much better than the second-best one. Let us sort the agents so that
	$b_1\, \mu_1 \geq b_2\, \mu_2 \geq \ldots \geq b_n\,\mu_n$.
The \emph{gap} $\delta$ of the problem instance is defined as
	$(b_1\, \mu_1 - b_2\, \mu_2)/\bmax$.
The \emph{$\delta$-gap regret} $R_\delta(T)$ is defined as the worst-case regret over all problem instances with gap $\delta$. Allocation rule \UCB{} achieves
	$R_\delta(T) = O(\tfrac{n}{\delta}\, \log T)$~\cite{bandits-ucb1};
there is a lower bound
	$R_\delta(T) = \Omega(\min(\tfrac{n}{\delta}\, \log T,\; \sqrt{nT}))$ \cite{Lai-Robbins-85,bandits-exp3,sleeping-colt08}.

\xhdr{\Realizations.}
A \emph{\realization} is a $n\times T$ table $\rho$ in which the $(i,t)$ entry $\rho_i(t)$ is the \payoff{} (e.g., the click bit) that agent $i$ receives if it is played in round $t$. Note that in order to fully define the behavior of any algorithm on all bid vectors one may need to specify all entries in the table, whereas only a subset thereof is revealed in any given run. We view $\rho$ as a realization of nature's random seed. Thus, we can now define ex-post truthfulness and other ex-post properties: informally, ex-post property is a property that holds for every given \realization.

For each agent $i$, round $t$, bid vector $b$ and \realization{} $\rho$, let $\mA_i^t(b;\rho)$  denote the probability that MAB allocation rule $\mA$ allocates the impression at round $t$ to agent $i$. (If $\ALG$ is deterministic, the probability $\ALG_i^t(\rho)$ is trivial: either $0$ or $1$.)

For MAB algorithm $\ALG$, define $\ALG_i^t(\rho)$ similarly.

\OMIT{ 
Likewise, we define \emph{ex-post regret}; to ensure that it generalizes~\refeq{eq:regret-stochastic}, we define it for a given \emph{distribution} $\D$ over \realizations:
\begin{align}\label{eq:regret-adversarial}
\Rex(T,b,\D)
	&\triangleq
	\textstyle{\max_i\, \E_{\rho \sim \D}[\, \sum_{t=1}^T b_i\, \rho_i(t)} \,] \nonumber \\
		&\qquad \qquad - \E[\, \textstyle{\sum_{t=1}^T}\; b_{i(t)}\,\rho_i(t) \,].
\end{align}
We will be interested in
	$\Rex(T) \triangleq \max R(T;b;\D)$,
where the maximum is taken over all distributions $\D$ and all bid vectors $b$ such that $b_i \leq 1$ for all $i$.

Again, regret guarantees for MAB algorithms easily translate to those on MAB allocation rules. In particular, allocation rule \EXP~\cite{bandits-exp3} achieves ex-post regret
	$\Rex(T) = O(\sqrt{nT \log n})$,
which is nearly optimal in view of the above-mentioned lower bound $R(T) = \Omega(\sqrt{nT})$.

For each agent $i$, round $t$, bid vector $b$ and \realization{} $\rho$, let $\mA_i^t(b;\rho)$ denote the probability that MAB allocation rule $\mA$ allocates the impression at round $t$ to agent $i$. Allocations $\mA$ and $\tilde{\mA}$ are \emph{distributionally equivalent} if
	$\mA_i^t(b;\rho) = \tilde{\mA}_i^t(b;\rho)$
for all $i,t,b,\rho$.
} 

\subsection{Truthfulness and monotonicity}
\label{sec:MAB-monotonicity}

\RefMainThm{} reduces the problem of designing truthful MAB mechanisms to that of designing monotone MAB allocations. Let us state this reduction explicitly:

\begin{theorem}\label{thm:MAB-truthful}
Consider the stochastic \MABprob. Let $\mA$ be a stochastically monotone (resp., ex-post monotone) MAB allocation rule. Applying the transformation in \RefMainThm\footnote{\RefMainThm{} is stated for the type space $T=(0,\infty)^n$, but it trivially extends to the case $T = (0,\bmax)^n$. } to $\mA$ with parameter $\mu$, we obtain a mechanism $\mathcal{M}$ such that:
\begin{itemize}
\item[(a)] $\mathcal{M}$ is {\bf\em stochastically truthful} (resp., {\bf\em ex-post truthful}), ex-post \noPositiveTransfers, and universally ex-post individually rational.
\item[(b)] for each \realization, the difference in expected welfare between $\mA$ and $\mathcal{M}$ is at most $\mu n T\,\bmax$.
\end{itemize}
\end{theorem}

Note that the theorem provides two distinct types of guarantees: game-theoretic guarantees in part (a), and performance guarantees in part (b).


We show that a very general class of deterministic MAB algorithms induces monotone MAB allocation rules (to which Theorem~\ref{thm:MAB-truthful} can be applied).

\begin{definition} \label{def:stats}
In a given run of an MAB algorithm, the \emph{round-$t$ statistics} is a pair of vectors $(\pi,\nu)$, where the $i$-th component of $\pi$ (resp., $\nu$) is equal to the total payoff (resp., the number  of impressions) of agent $i$ in rounds $1$ to $t-1$, for each agent $i$. Vectors $\pi$ and $\nu$ are called \emph{p-stats vector} and \emph{i-stats vector}, respectively.
\end{definition}

\begin{definition} \label{def:well-formed}
A deterministic MAB algorithm $\ALG$ is called \emph{\bf well-formed} if for each round $t$ and agent $i$, letting $(\pi,\nu)$ be the round-$t$ statistics, the following properties hold:
\begin{itemize}
\item {[$\ALG_i^t(\rho)$ is determined by $(\pi,\nu)$]}\hspace{2mm}
    there is a function $\chi_i(\pi;\nu)$ that depends only on the round-$t$ statistics such that
    	$\ALG_i^t(\rho) = \chi_i(\pi;\nu)$
for any \realization{} $\rho$ and all $t$.

\item {[$\chi$-monotonicity]} \hspace{2mm}
$\chi_i(\pi; \nu)$ is non-decreasing in $\pi_i$ for any fixed $(\pi_{-i},\nu)$.

\item {[$\chi$-IIA]} \hspace{2mm}
for each round $t$, any three distinct agents $\{i,j,l\}$ and any fixed $(\pi_{-i}, \nu_{-i})$, changing $(\pi_i, \nu_i)$ cannot transfer an impression from  $j$ to $l$.

\end{itemize}
\end{definition}

The {\em $\chi$-IIA} property above is reminiscent of \emph{Independence of Irrelevant Alternatives} (\emph{IIA}) property in the Social Choice literature (hence the name). A similar but technically different property is essential in the analysis of deterministic MAB allocation rules in~\cite{MechMAB-ec09}.

\OMIT { 
\item {[IIA]} for any three distinct agents $\{i,j,l\}$ and any fixed $(\pi_{-i}, \nu_{-i})$ the ratio between $p_j(\pi,\nu)$ and $p_l(\pi,\nu)$ is the same for all $(\pi_i,\nu_i)$ such that
    	$p_j(\pi,\nu)+p_l(\pi,\nu) >0$.
\begin{remark}
For deterministic well-formed algorithms, $p_i(\pi,\nu)\in \{0,1\}$ for all $i,\pi,n$. Moreover, given the monotonicity condition in Definition~\ref{def:well-formed}, the IIA is equivalent to the following:
\begin{itemize}
\item {\em [IIA$'$]} for each round $t$, any three distinct agents $\{i,j,l\}$ and any fixed $(\pi_{-i}, \nu_{-i})$, changing $(\pi_i, \nu_i)$ cannot transfer an impression from  $j$ to $l$.
\end{itemize}
\end{remark}
} 

\begin{remark}
For a concrete example of a well-formed MAB algorithm, consider (a version of) \UCB.\footnote{To ensure the $\chi$-IIA property, we use a slightly modified version of \UCB{}: $\log T$ is used instead of $\log t$,  and $\min$ is used to break ties (instead of an arbitrary rule). This change does not affect regret guarantees. We will denote this version as \UCB{} without further notice. } The algorithm is very simple: in each round $t$, it chooses agent
\begin{align*}
\min\left(
	\arg\max_i \left( \pi_i(t)/\nu_i(t) + \sqrt{8\log (T)/\nu_i(t)}
\right)\right).
\end{align*}

\end{remark}

\begin{lemma}\label{lm:well-formed}
In the stochastic \MABprob, let $\mA$ be a MAB allocation rule induced by a well-formed MAB algorithm. Then $\mA$ is stochastically monotone.
\end{lemma}

\OMIT{ 
\begin{theorem}\label{thm:well-formed}
Consider the stochastic \MABprob. Let $\mA$ be a MAB allocation rule induced by a well-formed MAB algorithm. Then $\mA$ is monotone. Thus, applying the transformation in \RefMainThm{} to $\mA$ (with parameter $\mu = \tfrac{1}{T}$) we obtain an MAB mechanism $\mathcal{M}$ such that:
\begin{itemize}
\item[(a)] $\mathcal{M}$ is {\bf\em truthful}, ex-post normalized, and universally ex-post individually rational,
\item[(b)] the social welfare of $\mathcal{M}$, and hence its regret is within an additive term of $\bmax$ from those of $\mA$ on any problem instance.
\end{itemize}
\end{theorem}

} 


\newcommand{\LmWellFormedClaim}{Claim~(\ref{eq:lm-well-formed-claim})}

\begin{proof}
We will use an alternative way to define a realization of random \payoffs: a \emph{stack-realization} is a $n\times T$ table in which the $(i,t)$ entry is the click bit that agent $i$ receives the $t$-th time she is played. Clearly a stack-realization and a bid vector uniquely determine the behavior of $\mA$. We will show that:
\begin{align}\label{eq:lm-well-formed-claim}
\text{$\mA$ is monotone for each stack-realization.}
\end{align}
Then $\mA$ is monotone in expectation over any distribution over stack-realizations, and in particular it is monotone in expectation over the random clicks in the stochastic MAB setting, so the Lemma follows.

\OMIT{ 
Fix stack-realization $\sigma$, agent $i$, and bid vector $b_{-i}$. Let $\nu_i(b_i,t)$ (resp., $\pi_i(b_i,t)$) be the total number of impressions (resp., the total \payoff) of agent $i$ in the first $t$ rounds, given a bid vector $(b_{-i},b_i)$.
Consider two bids $b_i<b^+_i$. We show by induction on $t$ that
	$\nu_i(b_i,t) \leq \nu_i(b^+_i,t)$
for all $t$. For the induction step we only need to worry about the case when the claim holds for a given $t$ with \emph{equality}. In this case
we show that
	$\nu_{-i}(b_i, t) = \nu_{-i}(b^+_i, t)$.
(This is trivial for $n=2$ agents; the general case requires a rather delicate argument that uses the $\chi$-IIA property.)
Since the \payoffs{} are coming from a fixed stack-realization, it follows that
	$\pi(b_i, t) = \pi(b^+_i, t)$.
Therefore, letting
 	$\hat{\pi}_i(b_i,t) = (b_i/\bmax)\; \pi_i(b_i,t)$.
be the total \emph{modified} \payoff{} of agent $i$, we have
	$\hat{\pi}_{-i}(b_i, t) = \hat{\pi}_{-i}(b^+_i, t)$
and
	$\hat{\pi}_i(b_i, t) < \hat{\pi}_i(b^+_i, t)$.
Thus, the claim follows by the $\chi$-monotonicity property in Definition~\ref{def:well-formed}.
} 

Let us prove \LmWellFormedClaim. Throughout the proof, fix stack-realization $\sigma$, agent $i$, and bid vector $b_{-i}$. Consider two bids $b_i<b^+_i$. The claim asserts that agent $i$ receives at least as many clicks with bid $b^+_i$ than with bid $b_i$.

Let us introduce some notation (letting $b_i$ be the bid of agent $i$). Let $\mA(b_i,t)$ be the agent selected by the allocation rule in round $t$. For each agent $j$, let $\nu_j(b_i,t)$ and $\pi_j(b_i,t)$ be, respectively, the total number of impressions and the total \payoff of agent $j$ in the first $t$ rounds. Let
	$\hat{\pi}_j(b_i,t) = (b_j/\bmax)\; \pi_j(b_i,t)$
be the corresponding total \emph{modified} \payoff. Let $\nu(b_i,t)$ (resp., $\pi(b_i,t)$ and $\hat{\pi}(b_i,t)$) be the $n$-dimensional vector whose $j$-th component is $\nu_j(b_i,t)$ (resp., $\pi_j(b_i,t)$ and $\hat{\pi}_j(b_i,t)$) for each agent $j$.

Note that $(\hat{\pi}(b_i,t),\nu(b_i,t))$ is the round-$t$ statistics for the MAB algorithm that $\mA$ is induced by. For each agent $j$,
$\nu_j(b_i, t)$ uniquely determines $\pi_j(b_i,t)$:
\begin{align}\label{eq:app-wellFormed-det-Ci}
	\pi_j(b_i, t) = \textstyle{\sum_{s=1}^{\nu_j}}\, \sigma(j,s)
	\quad\text{where}\quad
		\nu_j = \nu_j(b_i,t).
\end{align}


Let us overview the forthcoming technical argument. We will show by induction on $t$ that
	$\nu_i(b_i,t) \leq \nu_i(b^+_i,t)$
for all $t$. For the induction step we only need to worry about the case when the claim holds for a given $t$ with \emph{equality}. In this case
we show that
	$\nu_{-i}(b_i, t) = \nu_{-i}(b^+_i, t)$.
This is trivial for $n=2$ agents; the general case requires a rather delicate argument that uses the $\chi$-IIA property in Definition~\ref{def:well-formed}.\footnote{Also, we will use the fact that the probabilities $\chi_j(\hat{\pi},\nu)$ in Definition~\ref{def:well-formed} do not depend on the round (given $j$ and $(\hat{\pi},\nu)$). This is the only place in any of the proofs where we invoke this fact.}

Now let us carry out the proofs in detail.
First, denote
	$\nu_*(b_i,t) \triangleq t-\nu_i(b_i,t)$,
and let us show that for any two rounds $t,s$ it holds that
\begin{align}\label{eq:app-wellFormed-det-clm1}
\nu_*(b_i,t) = \nu_*(b^+_i,s)
	\;\Rightarrow\; \nu_{-i}(b_i,t) = \nu_{-i}(b^+_i,s).
\end{align}
Let us use induction on $\nu_*(b_i,t)$. For $\nu_*(b_i,t)=0$ the statement is trivial. For the induction step, suppose~\refeq{eq:app-wellFormed-det-clm1} holds whenever
	$\nu_*(b_i,t)=\nu_*$,
and let us suppose
	$\nu_*(b_i,t) = \nu_*(b^+_i,s) = \nu_*+1$.
Let $t'$ and $s'$ be the latest rounds such that
	$\nu_*(b_i,t') = \nu_*(b^+_i,s') =\nu_*$.
By the induction hypothesis,
	$\nu_{-i}(b_i,t') = \nu_{-i}(b^+_i,s')$.
It remains to prove that
	$\mA(b_i,t'+1) = \mA(b^+_i,s'+1)$,
i.e. that the allocation rule's selections in round $t'+1$ given bids $(b_{-i},b_i)$, and in round $s'+1$ given bids $(b_{-i},b^+_i)$, are the same.\footnote{Then $\nu_{-i}(b_i,t) = \nu_{-i}(b^+_i,s)$ because in all rounds from $t'+2$ to $t$ (resp., from $s'+2$ to $s$) agent $i$ is played.} By Definition~\ref{def:well-formed} these selections are uniquely determined (given the stack-realization) by the bids and the impression counts $\nu$. By the choice of $t'$ and $s'$, neither of the two selections is $i$, so by the $\chi$-IIA condition in Definition~\ref{def:well-formed} the selections are uniquely determined by $b_{-i}$ and $\nu_{-i}$, and hence are the same. This proves~\refeq{eq:app-wellFormed-det-clm1}.

Now, to prove \LmWellFormedClaim{} it suffices to show that for all $t$
\begin{align}\label{eq:app-wellFormed-det-clm2}
	\nu_i(b_i,t) \leq \nu_i(b^+_i,t).
\end{align}
Let us use induction on $t$. The claim is trivial for $t=1$, since the impression of agent $i$ in round $1$ does not depend on $(b;\sigma)$. For the induction step, assume that the assertion~\refeq{eq:app-wellFormed-det-clm2} holds for some $t$, and let us prove it for $t+1$. Note that (using the notation from Definition~\ref{def:well-formed})
\begin{align*}
	\nu_i(b_i,t+1) = \nu_i(b_i, t) + \chi_i(\hat{\pi}(b_i, t);\; \nu(b_i, t)).
\end{align*}
Now,
	$\nu_i(b_i, t) \leq \nu_i(b^+_i, t)$
by induction hypothesis. If the inequality is strict then~\refeq{eq:app-wellFormed-det-clm2} trivially holds for $t+1$. Now suppose
	$\nu_i(b_i, t) = \nu_i(b^+_i, t)$.
Then by~\refeq{eq:app-wellFormed-det-clm1} we have
	$\nu(b_i, t) = \nu(b^+_i, t)$.
Moreover, by~\refeq{eq:app-wellFormed-det-Ci} we have
	$\pi(b_i, t) = \pi(b^+_i, t)$
and therefore
	$\hat{\pi}_{-i}(b_i, t) = \hat{\pi}_{-i}(b^+_i, t)$
and
	$\hat{\pi}_i(b_i, t) < \hat{\pi}_i(b^+_i, t)$.
Thus, by the $\chi$-monotonicity property in Definition~\ref{def:well-formed} we have
\begin{align*}
	\chi_i(\hat{\pi}(b_i, t);\; \nu(b_i, t))
		\leq \chi_i(\hat{\pi}(b^+_i, t);\; \nu(b^+_i, t)).
\end{align*}
This concludes the proof of~\refeq{eq:app-wellFormed-det-clm2}, and that \LmWellFormedClaim.
\end{proof}

\subsection{Truthfulness and regret}
\label{sec:MAB-stochastic}

In this subsection we focus on the stochastic MAB setting, and consider the trade-off between regret and various notions of truthfulness. Ideally, one would like an MAB mechanism to be truthful in the strongest possible sense (universally ex-post), and have the same regret bounds as optimal MAB \emph{algorithms}.

Let us start with some background.  In Babaioff, Sharma and Slivkins~\cite{MechMAB-ec09} it was proved that any \emph{deterministic} mechanism that is ex-post truthful and ex-post normalized (under very mild restrictions), and any distribution over such deterministic mechanisms, incurs much higher regret than an optimal MAB algorithm such as \UCB. Namely, the lower bound in~\cite{MechMAB-ec09} states that
	$R(T) = \Omega(n^{1/3}\, T^{2/3})$,
whereas \UCB{} has regret
	$R(T) = O(\sqrt{n T \log T})$.~\footnote{Following the literature on regret minimization, we are mainly interested in the asymptotic behavior of $R(T)$ as a function of $T$ when $n$ is fixed.}
For $\delta$-gap instances the difference is even more pronounced: the analysis in~\cite{MechMAB-ec09} provides a polynomial lower bound of
    $R_\delta(T) = \Omega(\delta\,T^\lambda)$
for some $\lambda>0$, whereas \UCB{} achieves \emph{logarithmic} regret
	$R_\delta(T) = O(\tfrac{n}{\delta} \log T)$.

\OMIT{~\footnote{More precisely,
	$R_\delta(T) = \Omega(\delta\,T^\lambda)$
as long as the allocation rule has a non-trivial worst-case regret
    $R(T) =O(T^\gamma)$, $\gamma<1$
and
	$\lambda < 2(1-\gamma)$.}} 

Our first result is that we can use the machinery from Section~\ref{sec:MAB-monotonicity} to match the regret of \UCB{} for truthful mechanisms. We apply Theorem~\ref{thm:MAB-truthful} (with $\mu = \frac{1}{T}$) and Lemma~\ref{lm:well-formed} to \UCB{} to obtain the following corollary:

\begin{corollary}\label{cor:MAB-stochastic}
In the stochastic \MABprob, there exists a mechanism $\mathcal{M}$ such that
\begin{itemize}
\item[(a)] $\mathcal{M}$ is {\bf\em stochastically truthful}, ex-post \noPositiveTransfers, universally ex-post individually rational.

\item[(b)] $\mathcal{M}$ has regret
    	$R(T) = O(\sqrt{nT\log T})$
    and $\delta$-gap regret
    	$R_\delta(T) = O(\tfrac{n}{\delta} \log T)$.
\end{itemize}
\end{corollary}

\begin{remark}
The regret and $\delta$-gap regret in the above theorem are within small factors (resp., $O(\sqrt{\log T})$ and $O(1)$) of the best possible for any MAB allocation rule.
\end{remark}

\begin{remark}
~\cite{MechMAB-ec09} provides a weaker result which transforms any monotone MAB algorithm such as \UCB{} into a truthful and normalized MAB mechanism with matching regret bounds. The guarantees in~\cite{MechMAB-ec09} are weaker for the following reasons. First, it only applies to 0-1 \payoffs, whereas our setting allows for arbitrary \payoffs{} in $[0,1]$. Second, the individual rationality guarantee in~\cite{MechMAB-ec09} is much weaker: an agent may be charged more than her bid (which never happens in our mechanism), and the charge may be huge, as high as
	$b_i \times (4n)^T$;
thus, a risk-averse agent may be reluctant to participate. Third, the \noPositiveTransfers{}  guarantee is weaker: for some realizations of the \payoffs{} the expected payment may be negative. Finally, the payment rule in~\cite{MechMAB-ec09} requires (as stated) a prohibitively expensive computation.
\end{remark}

\OMIT{ 

Here we use Theorem~\ref{thm:well-formed} (applied to \UCB) to provide an MAB mechanism $\mathcal{M}$ with the regret bounds of \UCB{} and the better game-theoretic guarantees:
\begin{align}\label{eq:truthful-prop1}
\text{$\mathcal{M}$ is {\bf\em truthful}, ex-post normalized, and universally ex-post IR}.
\end{align}
Further, Theorem~\ref{thm:well-formed} can be applied to \EXP{} to guarantee the optimal \emph{ex-post} regret.

\begin{theorem}\label{thm:MAB-stochastic}
Consider the stochastic \MABprob.
\begin{itemize}
\item[(a)] [using \UCB] There exists an MAB mechanism $\mathcal{M}$ satisfying~\refeq{eq:truthful-prop1} with regret
    	$R(T) = O(\sqrt{nT\log T})$
    and $\delta$-gap regret
    	$R_\delta(T) = O(\tfrac{n}{\delta} \log T)$.
\item[(b)] [using \EXP] There exists an MAB mechanism $\mathcal{M}$ satisfying~\refeq{eq:truthful-prop1} with ex-post regret
     	$\Rex(T) = O(\sqrt{nT\log n})$.
\end{itemize}
\end{theorem}
} 

The truthfulness in Corollary~\ref{cor:MAB-stochastic} is only in expectation over the random \payoffs. Thus, after seeing a specific realization of the rewards an agent might regret having been truthful. Accordingly, we would like a stronger property: ex-post truthfulness, i.e. truthfulness \emph{for every given realization of the rewards}.

The main result of this section is an ex-post truthful MAB mechanism with optimal regret bounds. Unlike Corollary~\ref{cor:MAB-stochastic}, this result requires designing a new MAB allocation rule.\footnote{In particular, the allocation rule induced by \UCB\ is \emph{not} ex-post monotone and thus cannot be used to achieve ex-post truthfulness using the results of Section~\ref{sec:MAB-monotonicity}. To see that, consider a simple setting with two agents and two rounds, and a click realization in which both agents are not clicked at the first round, but are clicked at the second. With this click realization, an agent might be better off decreasing his bid in order to \emph{lose} (i.e., not be selected in) the first round, and then \emph{win} (i.e., be selected in) the second round.}
This allocation rule and its analysis are the main technical contributions.

\begin{theorem}\label{thm:MAB-exPost}
In the stochastic \MABprob, there is a mechanism $\mathcal{M}$ such that
\begin{itemize}
\item[(a)] $\mathcal{M}$ is {\bf\em ex-post truthful}, ex-post \noPositiveTransfers, and universally ex-post individually rational.
\item[(b)] $\mathcal{M}$ has regret
    	$R(T) = O(\sqrt{nT\log T})$
    and $\delta$-gap regret
    	$R_\delta(T) = O(\tfrac{n}{\delta} \log T)$.
\end{itemize}
\end{theorem}

The theorem follows from Theorem~\ref{thm:MAB-truthful} (with $\mu = \frac{1}{T}$) if there exists an MAB allocation rule that is ex-post monotone and has the claimed regret bounds.
Below we provide such allocation rule, called \newUCB.

\begin{lemma}
\newUCB{} is ex-post monotone and satisfies the regret bounds in Theorem~\ref{thm:MAB-exPost}(b).
\end{lemma}

\begin{remark}
\newUCB{} is deterministic. While not essential for Theorem~\ref{thm:MAB-exPost}, this fact confirms the intuition from~\cite{MechMAB-ec09} that the main obstacle for deterministic ex-post truthful MAB mechanisms is insufficient observable information to compute payments rather than ex-post monotonicity of an allocation rule.
\end{remark}

\newUCB{} maintains a set of active agents; initially all agents are active. For each round $t$, there is a \emph{designated agent}
	$i = 1+ (t \bmod{n})$.
If this agent is active, then it is allocated. Else, an active agent is chosen (according to some fixed ordering on the agents) and allocated. For each agent $i$, lower and upper confidence bounds $(L_i,U_i)$ on the product $b_i\, \mu_i$ are maintained (recall that $\mu_i$ is the CTR of agent $i$). After each round, each agent is de-activated if its
upper confidence bound is smaller than someone else's lower confidence bound. The pseudocode is in Algorithm~\ref{alg:ex-post-UCB}.

\begin{algorithm}[h]
\caption{\newUCB: ex-post monotone MAB allocation rule.}
\label{alg:ex-post-UCB}
\begin{algorithmic}[1]
\STATE {\bf Given:} $n=\text{\#agents}$, $T = \text{\#rounds}$,
	upper bound $\bmax$.
\STATE Solicit a bid vector $b$ from the agents;
	$b\leftarrow b/\bmax$.
\STATE {\bf Initialize:} set of active agents
	$\Active = \{\text{all agents}\} $.
\FORALL{agent $i$}
	\STATE $c_i \leftarrow 0$; $n_i \leftarrow 0$
		\COMMENT{total click reward and \#impressions}
	\STATE \hspace{3mm}
		\COMMENT{the totals are only over ``designated" rounds}
	\STATE $U_i \leftarrow b_i$; $L_i \leftarrow 0$
		\COMMENT{Upper and Lower Confidence Bounds}
\ENDFOR
\STATE \COMMENT{Main Loop}
\FOR{rounds $t=1,2,\,\ldots,\,T$}
	\STATE $i \leftarrow 1+(t \bmod{n})$.
		\COMMENT{The ``designated" agent}
	\IF{$i\in \Active$}
	\STATE Allocate agent $i$.
	\STATE $n_i \leftarrow n_i+1$;
		   $c_i \leftarrow c_i + \mathtt{reward}$.
		   \COMMENT{Update statistics.}
	\STATE \COMMENT{Update confidence bounds.}
	\IF{$L_i<U_i$}
		\STATE $(L'_i,\; U'_i) \leftarrow
			b_i\,(c_i/n_i \mp \sqrt{8\log(T)/n_i}))$.
		\IF{$\max(L_i,\, L'_i)< \min(U_i,\, U'_i)$}
			\STATE $(L_i,\; U_i) \leftarrow
				(\max(L_i,\, L'_i),\; \min(U_i,\, U'_i))$.
		\ELSE
		\STATE $(L_i,\; U_i) \leftarrow
					(\tfrac{L_i+U_i}{2},\; \tfrac{L_i+U_i}{2})$.
		\ENDIF
	\ENDIF
	\ELSE
		\STATE Allocate agent $i = \min \Active$.
	\ENDIF
	\FORALL{agent $i\in \Active$}
		\IF{$U_i < \max_{j\in \Active} L_j$}
		\STATE Remove $i$ from $\Active$.
		\ENDIF
	\ENDFOR
\ENDFOR
\end{algorithmic}
\end{algorithm}

\newcommand{\DRule}{(\ensuremath{\mathtt{rule}})} 

Fix realization $\rho$ and bid vector $b$. Let $\Active(t,b)$ be the set of active agents \asedit{in the beginning of round $t$}. For each agent $i$, let $L_i(t,b)$ and $U_i(t,b)$ be the values of $L_i$ and $U_i$ in the end of round $t$. 

The goal of the specific update rules for the confidence bounds (lines 15-20) and the statistics (lines 13) is to guarantee the following two properties:
\begin{itemize}
\item the statistics are kept only for rounds when a designated agent is played. Moreover, for each agent $i$ and round $t$, and any two bid vectors $b$ and $b'$ we have
\begin{align}\label{eq:thm-exPostMab-prop2}
\text{if } i\in \Active(t,b) \cap \Active(t,b') \text{ then }
 	\left\{\begin{array}{rcl}
	 	L_i(t,b)/b_i &=& L_i(t,b')/b'_i  \\
		U_i(t,b)/b_i &=& U_i(t,b')/b'_i
	\end{array}\right..
\end{align}

\item for any fixed realization $\rho$ and bid vector $b$, and each agent $i$: $L_i\leq U_i$, and from round to round $L_i$ is non-decreasing and $U_i$ is non-increasing. In other words, for each round $t$ it holds that
    \begin{align}\label{eq:thm-exPostMab-prop1}
		L_i(t-1,b) \leq L_i(t,b) \leq U_i(t,b) \leq U_i(t-1,b).
	\end{align}
\end{itemize}

\noindent The ex-post monotonicity follows from these two properties and the de-activation rule (lines 24-25).

\xhdr{Ex-post monotonicity.}
Let
	$L^*(t,b) \triangleq \max_{i\in \Active(t,b)} L_i(t,b)$.
Fix agent $i$ and $b^+_i> b_i$, and let
	$b^+ = (b_{-i},\,b^+_i)$
be the ``alternative" bid vector. Let $\lambda = b_i/b^+_i$.

\begin{claim}
We establish the following sequence of claims:
\begin{OneLiners}
\item[(C1)] $L^*(t,b)$ is non-decreasing in $t$, for any fixed $b$.
\item[(C2)] For each round $t$, $L^*(t,b) \leq \asedit{\lambda}\, L^*(t,b^+)$.
\item[(C3)] For each round $t$,
    $\Active(t,b^+) \setminus \{i\} \subset \Active(t,b)\setminus \{i\}$.
\item[(C4)] In each round $t$: if
	$i\in \Active(t,b)$
then
	$i\in \Active(t,b^+)$.
\end{OneLiners}
\end{claim}

\begin{proof}
Let us prove the parts (C1-C4) one by one.

\xhdr{(C1).}
We use~property~\eqref{eq:thm-exPostMab-prop1} and the de-activation rule.  Throughout the proof, we omit the bid vector $b$ from the notation. Fix round $t\geq 2$. Let $i\in \Active(t-1)$ be an agent such that
	$L^*(t-1)= L_i(t-1)$.
If $i\in \Active(t)$ then
\begin{align*}
	L^*(t-1) = L_i(t-1) \leq L_i(t) \leq L^*(t)
\end{align*}
Else $i$ is de-activated in round $t$, so
\begin{align*}
	L^*(t-1) = L_i(t-1) \leq L_i(t) \leq U_i(t) < L^*(t).
\end{align*}

\xhdr{(C2).}
Suppose, for the sake of contradiction, that
	$L^*(t,b) > \asedit{\lambda}\, L^*(t,b^+)$.
Let $j\in \Active(t,b)$ be an agent such that
	$L_j(t,b) = L^*(t,b)$.
\asedit{If $j\in \Active(t,b^+)$ then by property~\eqref{eq:thm-exPostMab-prop2} we have
\begin{align*}
L^*(t,b) = L_j(t,b) = \lambda\,L_j(t,b^*) \leq \lambda\,L^*(t,b^+),
\end{align*}
contradiction.} We conclude that
	$j\not\in \Active(t,b^+)$.
Thus with bid vector $b^+$ agent $j$ gets disqualified during some round $s<t$. Thus, \begin{align}\label{eq:pf-cl-CC}
U_j(s,b^+) < L^*(s,b^+) \leq L^*(t,b^+),
\end{align}
where the second inequality is by Part (C1).
Now using~property~\eqref{eq:thm-exPostMab-prop1} and property~\eqref{eq:thm-exPostMab-prop2} (for the right-most inequality), we get that
\begin{align*}
L^*(t,b)
	= L_j(t,b)
	\leq U_j(t,b)
	\leq U_j(s,b)
	\asedit{=\lambda}\, U_j(s,b^+).
\end{align*}
Thus, $L^*(t,b) \leq \asedit{\lambda}\, L^*(t,b^+)$ by \refeq{eq:pf-cl-CC}, the desired contradiction.

\xhdr{(C3).}
Use induction on $t$. The claim trivially holds for $t=1$.
Assuming the claim holds for some $t$ we prove it holds for $t+1$.
Fix agent
	$j \in \Active(t\asedit{+1},b^+) \setminus \{i\}$.
We need to prove that $j \in \Active(t+1,b)$.

Note that $j\in \Active(t,b^+)$, and so $j\in \Active(t,b)$ by the induction hypothesis. \asedit{Therefore
\begin{align*}
L^*(t,b)
    &\leq \lambda\, L^*(t,b^+)
        &\text{(by Part (C2))} \\
    &\leq \lambda\, U_j(t,b^+)
        &\text{(by the de-activation rule)} \\
    &= U_j(t,b)
        &\text{(by property~\eqref{eq:thm-exPostMab-prop2})}
\end{align*}
So agent $j$ is not deactivated in round $t$ under bid vector $b$, i.e.
    $j \in \Active(t+1,b)$,
completing the proof.}

\xhdr{(C4).}
Use induction on $t$. The base case $t=0$ holds because initially all agents are active. For the induction step, assume that the statement holds for some round $t\geq 0$. Suppose $i\in \Active(t+1,b)$. We need to prove that $i\in \Active(t+1,b^+)$.

\asedit{Note that $i\in \Active(t,b)$, and so $i\in \Active(t,b^+)$ by the induction hypothesis. Therefore
\begin{align*}
L^*(t,b)
    &\leq U_i(t,b)
        &\text{(by the de-activation rule)} \\
    &=\lambda\, U_i(t,b^+)
        &\text{(by property~\eqref{eq:thm-exPostMab-prop2})}.
\end{align*}

If $L^*(t,b) = \lambda\, L^*(t,b^+)$, then
    $L^*(t,b^+) = L^*(t,b)/\lambda \leq U_i(t,b^+)$,
and we are done.

From here on, assume $L^*(t,b) \neq \lambda\, L^*(t,b^+)$. Then
    $L^*(t,b) < \lambda\,L^*(t,b^+)$
by Part (C2). Pick agent $j$ which maximizes $L_j(t,b^+)$.
If $j\neq i$ then $j\in \Active(t,b)$ by Part (C3), so
\begin{align*}
L^*(t,b)
    &\geq L_j(t,b)
        & \\
    &= \lambda\, L_j(t,b^+)
        &\text{by property~\eqref{eq:thm-exPostMab-prop2})} \\
    &= \lambda\, L^*(t,b^+)
        &\text{(by the choice of $j$)},
\end{align*}
contradicting our assumption. Then $j=i$, and so
$ L^*(t,b^+) = L_i(t,b^+) \leq U_i(t,b^+)$.}
\end{proof}

Ex-post monotonicity follows easily from (C3-C4).

\begin{claim}
Consider a fixed round $t$. Suppose agent $i$ is allocated with bid vector $b$. Then it is also allocated with bid vector $b^+$.
\end{claim}

\begin{proof}
Since $i\in\Active(t,b)$, by (C4) we have $i\in \Active(t,b^+)$.

If agent $i$ is the designated agent in round $t$, then as such it is allocated under both bid vectors. If agent $i$ is \emph{not} the designated agent in round $t$, then
    $i = \min \Active(t,b)$.
By (C3) it holds that
    $\Active(t,b^+) \subset \Active(t,b)$,
which implies that
    $i = \min \Active(t,b^+)$.
So $i$ is allocated under bid vector $b^+$, too.
\end{proof}

\OMIT{ 
Fix round $t$ and let $q_i(t,b)$ be the probability that with bid vector $b$, agent $i$ receives an impression in this round.
\begin{claim}
For each round $i$,
	$q_i(t,b) \leq q_i(t,b^+)$.
\end{claim}
\begin{proof}
If $q_i(t,b)>0$ then agent $i$ is active with bid vector $b$, and hence (using (C4)) with $b^+$. If agent $i$ is the designated agent in round $t$, then
	$q_i(t,b) = q_i(t,b^+)=1$.
Else,
	$q_i(t,\,\cdot) = 1/|\Active(t,\,\cdot)|$,
and so
	$q_i(t,b) \leq q_i(t,b^+)$
since by (C3)
	$|\Active(t,b^+)| \leq |\Active(t,b)| $.
\end{proof}
} 

\xhdr{Regret analysis.}
The regret analysis is relatively standard, following the ideas in~\cite{bandits-ucb1}. For simplicity assume that $\bmax=1$. Fix a bid vector $b$.

For each agent $i$, let $c_i(t)$ and $n_i(t)$ be, respectively, the number of clicks and impressions in all rounds $s\leq t$ when it is allocated as the designated agent. Let
	$r_i(t) = \sqrt{8 \log(T)/n_i(t)}$.
Then the event
\begin{align}\label{eq:regret-highProb-event}
	|\mu_i - c_i(t)/n_i(t)| \leq r_i(t) \quad \text{for all rounds $t$}
\end{align}
holds with probability at least $1-T^{-2}$.~%
\footnote{This follows from Azuma-Hoeffding inequality via a standard argument, one version of which we provide below. Fix agent $i$. For each $s\in N$, let $X_s$ be the click bit for the $s$-th time this agent is allocated as the designated agent, if $s\leq n_i(T)$, and otherwise define $X_s$ to be an independent 0-1 random variable with expectation $\mu_i$. Then the random variables $Y_s=X_s-\mu_s$, $s\in \N$ form a martingale. Applying Azuma-Hoeffding inequality to $Y_1 ,\, \ldots\,, Y_N$, for any given $N$, we obtain that the event
    $|\sum_{s=1}^N Y_s| \leq \sqrt{8N\log(T)}$
holds with probability at least $1-T^{-3}$. Taking the Union Bounds over all $N\leq T$, and noting that
    $c_i(t) = \sum_{s=1}^{n_i(t)} X_s $,
it follows that the event \eqref{eq:regret-highProb-event} holds with probability at least $1-T^{-2}$.}
In what follows, let us assume that this event holds for all agents $i$. (The regret accumulated if this event fails is negligible.)

Then it easily follows from the specs of \newUCB{} that for each agent $i$,
\begin{align*}
\left\{\begin{array}{l}
	L_i(t,b) \leq b_i\,\mu_i \leq U_i(t,b) \\
	U_i(t,b) - L_i(t,b) \leq 2\,r_i(t)    
\end{array}\right.
\end{align*}

Let
    $i^* \in \argmax_i b_i \mu_i$
be a best agent. Note that
    $U_{i^*}(t,b) \geq b_{i^*} \mu_{i^*} \geq b_i \mu_i \geq L_i(b,t)$
for all agents $i$ and rounds $t$. It follows that $i^*$ is never de-activated by the algorithm.

Consider some agent $i$ with
	$\Delta_i \triangleq b_{i^*}\mu_{i^*}  - b_i\,\mu_i>0$.
Then $r_i(t)<\Delta_i$ after $O(\Delta_i^{-2} \log T)$ rounds in which this agent is allocated as the designated agent. After such round $t$,
\[ U_i(b,t)  \leq b_i\mu_i + r_i(t)
        < b_i\mu_i + \Delta_i
         = b_{i^*}\mu_{i^*}
         \leq L_{i*}(b,t), \]
and therefore agent $i$ is deactivated. It follows that agent $i$ is allocated as the designated agent at most $O(\Delta_i^{-2} \log T)$ times. Therefore it is de-activated after at most $O(k\,\Delta_i^{-2} \log T)$ rounds. This, in turn, implies the claimed regret bound.

\subsection{The power of randomization}
\label{sec:separation}

A by-product of Theorem~\ref{thm:MAB-exPost} is a separation between the power of  deterministic and randomized mechanisms, in terms of regret for MAB mechanisms that are ex-post truthful and ex-post normalized. The lower bound for deterministic mechanisms is from~\cite{MechMAB-ec09}.

One challenge here is to ensure that the upper and lower bounds talk about exactly the same problem; as stated, Theorem~\ref{thm:MAB-exPost} and the main lower bound result from~\cite{MechMAB-ec09} do not. To bypass this problem, we focus on the case of two agents, and use a more general version of the lower bound: Theorem C.1 in the full version of~\cite{MechMAB-ec09}. Further, to match~\cite{MechMAB-ec09} we extend the mechanism from Theorem~\ref{thm:MAB-exPost} to a setting in which $\bmax$ is not known a priori.

We formulate the separation theorem as follows. Denote
 	$R(T,\bmax) \triangleq \max R(T;b;\mu)$,
where the maximum is taken over all CTR vectors $\mu$ and all bid vectors $b$ such that $b_i\leq \bmax$ for all $i$.

\begin{theorem}\label{thm:MAB-separation}
Consider the stochastic \MABprob{} with two agents. Assume  $\bmax$ is not known a priori to the mechanism. Suppose $\mathcal{M}$ is an MAB mechanism that is (i) ex-post truthful and ex-post normalized,
and (ii) has regret
  $R(T,\bmax) = \tilde{O}(\bmax\, T^\gamma)$
for some $\gamma$ and any $\bmax$. Then:
\begin{itemize}
\item[(a)] \cite{MechMAB-ec09}  If $\mathcal{M}$ is deterministic then $\gamma\geq\tfrac23$.
\item[(b)] There exists such \emph{randomized} $\mathcal{M}$ with $\gamma=\tfrac12$.
\end{itemize}
\end{theorem}

\begin{proof}[of part (b)]
Let $\mA'$ be the ex-post monotone MAB allocation rule in Theorem~\ref{thm:MAB-exPost}, for $\bmax=1$. Define an MAB allocation rule $\mA$ as a rule that inputs the bid vector $b$ and passes the modified bid vector
	$b' = b/ (\max_i b_i)$
to $\mA'$. We claim that $\mA$ is ex-post monotone, too. Indeed, w.l.o.g. assume $b_1>b_2$. If $b_2$ increases (to a value $\leq b_1$), then $b'_2$ increases while $b'_1$ stays the same. Thus, the total \payoff{} of agent $2$ increases. If $b_1$ increases then $b'_2$ decreases while $b'_1$ stays the same, so the total \payoff{} of agent $2$ does not increase, which implies that the total \payoff{} of agent $1$ does not decrease. Claim proved. Now part (b) follows from Theorem~\ref{thm:MAB-truthful}  (with $\mu = \frac{1}{T}$).
\end{proof}


\OMIT { 
Theorem~\ref{thm:well-formed} follows from the following two lemmas, using \RefMainThm{} and Theorem~\ref{thm:improved-approximation} (the latter is needed to bound the loss in welfare). Note that Lemma~\ref{lm:wellFormed-equiv} is not needed if the well-formed allocation rule in question is deterministic.

\begin{lemma}\label{lm:wellFormed-equiv}
Let $\mA$ be a well-formed (randomized) MAB algorithm. Then $\mA$ is distributionally equivalent to a distribution over deterministic well-formed MAB algorithms.
\end{lemma}

\begin{lemma}\label{lm:wellFormed-det}
Let $\mA$ be a MAB allocation rule induced by a deterministic well-formed MAB algorithm. Then $\mA$ is monotone.
\end{lemma}

In the interest of space, we only sketch the proofs. The full proofs are in Appendix~\ref{app:wellFormed-equiv} and Appendix~\ref{app:wellFormed-det}, respectively.

\begin{proof}[Proof Sketch of Lemma~\ref{lm:wellFormed-equiv}]
Fix round $t$. In any given run of $\mA$, the algorithm chooses an agent $i$ from the distribution given by $p_1(\pi,\nu) , \,\ldots\, , p_n(\pi,\nu)$, where
$(\pi,\nu)$ is the round-$t$ statistics. Thus, in any given round the algorithm makes such choice for a single pair $(\pi,\nu)$. We need to choose a single random seed $r^t$ that would deterministically define the selection for any given pair $(\pi,\nu)$. We define
	$r^t = (r_1 ,\,\ldots\,,r_k)$,
where each $r_i$ is a number selected independently, uniformly at random from the interval $(0,1)$. Given a specific pair $(\pi,\nu)$, this random seed is interpreted as follows: select agent
\begin{align*}
	i^*(\pi,\nu) = \arg\min_i \ln(1-r_i)/p_i(\pi,\nu).
\end{align*}
This rule relies on a well-known fact about exponential distributions which states that
	$\Pr[i^*(\pi,\nu) = i] = p_i(\pi,\nu)$
for each agent $i$. Each vector of random seeds
	$(r^1, \,\ldots\,, r^T) $
defines a deterministic MAB algorithm $\mA^*$. In Appendix~\ref{app:wellFormed-equiv} we use the IIA property of $\mA$ to show that any such $\mA^*$ satisfies the monotonicity and IIA properties in Definition~\ref{def:well-formed}.
\end{proof}
} 

\section{Extension to multi-parameter domains}
\label{sec:multiParam}

\renewcommand{\type}{\mathbf{x}}
\newcommand{\bid}{\mathbf{b}}

Our general transformation from Section~\ref{sec:generic} can be extended to multi-parameter mechanisms.%
It is known that a multi-parameter allocation rule is truthfully implementable if and only if it satisfies a property called ``cycle-monotonicity''. Similar to the single-parameter case, we present a general procedure to take any cycle-monotone allocation rule $\mA$ and transform it into a randomized mechanism that is truthful-in-expectation, implements the same outcome as $\mA$ with probability arbitrarily close to $1$, and requires evaluating that allocation rule only once. The technical contribution here is that we find a reduction from the multi-parameter setting to the single-parameter case.

This section is self-contained. For more background on multi-parameter mechanisms for a CS-oriented audience, please refer to \cite{ArcherKleinberg-sigecom08,ArcherKleinberg-ec08}. An Economics-oriented background for this area can be found in~\cite{Ashlagi-econometrica10}.

\subsection{Preliminaries: multi-parameter domains}
\label{sec:multiParam-prelims}

\xhdr{Generalized types.}
In the full generality, multi-parameter mechanisms are defined as follows. There are $n$ agents and a set $\outcomes$ of outcomes. Each agent $i$ is characterized by his \emph{type}
    $\type_i: \outcomes \to \Re$,
where $\type_i(o)$ is interpreted as the agent's valuation for the outcome $o\in\outcomes$. For each agent $i$ there is a set of feasible types, denoted $\types_i$. Denote
    $\types = \types_1 \times \ldots \times \types_n$
and call it the \emph{type space};  call $\types_i$ the type space of agent $i$. The mechanism knows $(n,\outcomes,\types)$, but not the actual types $\type_i$;  each type $\type_i$ is known only to the corresponding agent $i$. Formally, a problem instance, also called a \emph{multi-parameter domain}, is a tuple $(n,\outcomes,\types)$.

Using this general notion of types, we define truthful mechanisms in essentially the same way as in Section~\ref{sec:prelims}, with minimal syntactic changes. A (direct revelation) mechanism $\mathcal{M}$ consists of the pair $(\mA, \mP)$,
where $\mA:\types\rightarrow \outcomes$ is the {\em allocation rule} and $\mP:\types\rightarrow \Re^n$ is the {\em payment rule}. Both $\mA$ and $\mP$ can be randomized. Each agent $i$ reports a type $\bid_i\in \types_i$ to the mechanism, which is called the {\em bid} of this agent. We denote the vector of bids by $\bid = (\bid_1 \LDOTS \bid_n) \in \types$. The mechanism receives the bid vector $\bid\in\types$, selects an outcome $\mA(\bid)$, and charges each agent $i$ a payment of $\mP_i(\bid)$.
The utilities are quasi-linear and agents are risk-neutral: if agent $i$ has type $\type_i \in \types_i$ and the bid vector is $\bid\in \types$, then this agent's  utility is
\begin{align}\label{eq:multiParam-util-defn}
u_i(\type_i;\bid) = \E_{\mathcal{M}} \left[ \type_i(\mA(\bid)) - \mP_i(\bid) \right].
\end{align}

For each type $\type_i\in \types_i$ of agent $i$ we use a standard notation $(\bid_{-i},\type_i)$ to denote the bid vector $\hat{\bid}$ such that $\hat{\bid}_i=\type_i$ and $\hat{\bid}_j = \bid_j$ for every agent $j\neq i$.

\xhdr{Special case: dot-product valuations.}
For intuition, consider \emph{dot-product valuations}, an important special case where the type $\type\in \types_i$ of each agent $i$ can be decomposed as a dot product
    $\type(o) = \beta_{\type}\cdot a_i(o)$,
for each outcome $o\in\outcomes$, where
    $\beta_{\type}, a_i(o)\in \Re^{d}$
are some finite-dimensional vectors.
Here the term $a_i(o)$ is the same for all types $\type\in \types_i$ (and known to the mechanism), whereas $\beta_{\type}$ is the same for all outcomes $o\in \outcomes$ and is known only to agent $i$. The term $a_i(o)$ is usually called an ``allocation'' of agent $i$ for outcome $o$, and $\beta_{\type}$ is called the ``private value''. The single-parameter domains defined in Section~\ref{sec:prelims} correspond to the case $d = 1$.

Note that the type $\type$ of each agent $i$ is determined by the corresponding private value $\beta_{\type}$, and his type space $\types_i$ is determined by
    $D_i = \{ \beta_{\type}:\; \type\in\types_i\} \subset \Re^d$.
Because of this, in the literature on dot-product valuations the term ``type'' often refers to $\beta_{\type}$. To avoid ambiguity, in this section we will refer to $\beta_{\type}$ as ``private value'' rather than ``type'', and call $D_1\times \ldots \times D_n$ the \emph{private value space}.

\xhdr{Game-theoretic properties.}
Truthfulness and individual rationality are defined exactly as in Section~\ref{sec:prelims} if expressed in terms of the agents' utility:
\begin{itemize}
\item A mechanism is {\em truthful} if for every agent $i$ truthful bidding is a {\em dominant strategy}:
\begin{equation}\label{eq:multiParam-truthful}
u_i(\type_i; (\bid_{-i}, \type_i))
    \geq u_i(\type_i; \bid)
    \quad \forall \type_i\in \types_i, \;\bid \in \types.
\end{equation}
An allocation rule is called \emph{truthfully implementable} if it is the allocation rule in some truthful mechanism.

\item A mechanism is {\em individually rational (IR)} if each agent $i$  never receives negative utility by participating in the mechanism and bidding truthfully:
    \begin{equation}\label{eq:multiParam-IR}
        u_i(\type_i; (\bid_{-i}, \type_i)) \geq 0
            \quad \forall \type_i\in \types_i, \;\bid_{-i} \in \types_{-i}.
    \end{equation}
\end{itemize}

The right-hand side in~\eqref{eq:multiParam-IR} represents the maximal guaranteed utility of an ``outside option'' (i.e., from not participating in the mechanism). For example, our definition of IR is meaningful whenever this utility is $0$, which is a typical assumption for most multi-parameter domains studied in the literature.

\xhdr{Our assumptions.}
We make two assumptions on the type space $\types$:
\begin{itemize}
\item non-negative types: $\type_i(o)\geq 0$ for each agent $i$, each type $\type_i\in \types_i$, and each outcome $o\in\outcomes$.

\item \emph{rescalable types}: $\lambda \type_i\in \types_i$
for each agent $i$, each type $\type_i\in \types_i$, and any parameter $\lambda\in[0,1]$.

\OMIT{\item The type space is \emph{convex}: for each agent $i$, any types $\type, \type'\in \types_i$, and any parameter $\lambda\in[0,1]$ it holds that
    $\lambda \type + (1-\lambda) \type' \in \types_i $.}
\end{itemize}

\noindent For dot-product valuations, types are rescalable if and only if it holds that
    $\beta_{\type}\in D_i \Rightarrow \lambda \beta_{\type}\in D_i $
for each $\lambda \in [0,1]$. Thus, assuming rescalable types is equivalent to assuming that the set $D_i$ is star-convex at $0$. To ensure non-negative types, it suffices to assume that
    $D_i\subset \Re_+^{d}$ for each agent $i$, and all allocations are non-negative: $a_i(o) \in \Re_+^{d}$ for all $o \in \outcomes$.

In particular, for each agent $i$ there exists a \emph{zero type}: a type $\type_i\in \types_i$ such that $\type_i(\cdot)\equiv 0$. Let us say that a mechanism is \emph{normalized} if for each agent $i$, the expected payment of this agent is $0$ whenever she submits the zero type.

\xhdr{Truthfulness characterization.}
We will use the following characterization of truthful mechanisms. A (randomized) allocation rule $\mA$  is \emph{cycle-monotone} if the following property holds: for each bid vector $\bid\in\types$, each agent $i$, each $k\geq 2$, and each $k$-tuple
    $\type_{i,0},\; \type_{i,1} \LDOTS \type_{i,k} \in \types_i$
of this agent's types, we have
\begin{align}\label{eq:CMON}
\E_{\mA}\left[ \sum_{j=0}^k
    \type_{i,j}\left( o_{i,j} \right)
    -     \type_{i,\,(j-1) \bmod{k}}\left( o_{i,j} \right)
    \right] \geq 0,\quad
    \text{ where }
    o_{i,j} = \mA\left(\bid_{-i},\, \type_{i,j} \right)\in \outcomes.
\end{align}

\begin{theorem}[\citeN{Rochet-1987}]
Consider an arbitrary multi-parameter domain $(n,\outcomes,\types)$. A (randomized) allocation rule $\mA$ is truthfully implementable if and only if it is cycle-monotone. Assuming rescalable types, for any cycle-monotone allocation rule $\mA$, a mechanism $(\mA,\mP)$ is truthful and normalized if and only if
\begin{align}\label{eq:multiParam-myerson}
\E_{\mA}\left[ \mP_i(\bid) \right] =
   \E_{\mA}\left[
    \bid_i(\mA(\bid)) -
     \int_{t=0}^1  \bid_i(\mA(\bid_{-i},\, t\, \bid_i )) \, dt
    \right].
\end{align}
\end{theorem}

Note that this theorem generalizes Theorem~\ref{thm:myerson} for single-parameter mechanisms, as applied to single-parameter domains with private value space $[0,1]^n$. In particular,~\refeq{eq:multiParam-myerson} generalizes the Myerson payment rule for single-parameter mechanisms.

\subsection{The multi-parameter transformation}

Consider allocation rule $\mA$, bid vector $\bid\in \types$, and the rescaling vector $\lambda \in [0,1]^n$. Denote
    $$\lambda \otimes b = (\lambda_1 \bid_1 \LDOTS \lambda_n \bid_n ) \in \types. $$
In other words, $\lambda\otimes b$ is the ``rescaled'' bid vector where the bid of each agent $i$ is $\lambda_i \bid_i$; this bid vector is well-defined because we assumed the rescalable types property. Note that for each $\bid$ the subset
    $$\types_{\bid} = \{ \lambda\otimes b:\; \lambda \in [0,1]^n\} \subset \types$$
forms a single-parameter type space where each agent $i$ has private value $\lambda_i\in [0,1]$ and allocation $b_i(o)$ for every outcome $o$. By abuse of notation, let us treat the allocation / payment rules for $\types_{\bid}$ as functions from the private value space $[0,1]^n$ rather than the type space $\types_{\bid}$.

Consider an allocation rule
   $\mA_{\bid}(\lambda) = \mA(\lambda\otimes \bid)$
for the single-parameter type space $\types_{\bid}$.
If the original allocation rule $\mA$ is truthfully implementable for type space $\types$ using payment rule $\mP$, then $\mA_{\bid}$ is truthfully implementable for type space $\types_{\bid}$ using payment rule
   $\mP_{\bid}(\lambda) = \mP(\lambda\otimes \bid)$,
because restricting the allocation and payment rules to $\types_{\bid}$ only
limits the set of possible misreports of an agent. Essentially, the idea will be to apply our single-parameter transformation to $\mA_{\bid}$.

Let $\mathbf{f}_{\mu} = (f_1 \LDOTS f_n)$ be the $n$-tuple of canonical self-resampling procedures with resampling probability $\mu$, for some fixed $\mu\in (0,1)$. (See Algorithm~\ref{alg:canon-proc} on page~\pageref{alg:canon-proc}.) For each bid vector $\bid\in\types$, let
\begin{align}\label{eq:multiParam-reduction-1}
 (\tA_{\bid},\tP_{\bid}) = \GenericMech(\mA_{\bid},\mu,\mathbf{f}_{\mu})
\end{align}
be the single-parameter mechanism for type space $\types_{\bid}$ obtained by applying our single-parameter transformation from Section~\ref{sec:generic} to allocation $\mA_{\bid}$.%
\footnote{Note that the transformed mechanism depends on $\mu$. We do not make this dependence explicit, to simplify the notation.}

The transformed multi-parameter mechanism is defined as
\begin{align}\label{eq:multiParam-reduction}
    \left( \tA(\bid),\; \tP(\bid) \right) = \left( \tA_{\bid}(\vec{1}\,),\; \tP_{\bid}(\vec{1}\,) \right)
        \text{ for every } b\in \types.
\end{align}
This completes the description of our multi-parameter transformation. The useful properties of this transformation are captured in the theorem below.

\begin{theorem}\label{thm:multiParam}
Consider an arbitrary multi-parameter domain $(n,\outcomes,\types)$ with rescalable, non-negative types. Let $\mA$ be a cycle-monotone allocation rule. Let $\mathcal{M}_\mu= (\tA,\tP)$ be the transformed mechanism defined by Equations(\ref{eq:multiParam-reduction-1}-\ref{eq:multiParam-reduction}), for some parameter $\mu\in(0,1)$. Then $\mathcal{M}_\mu$ has the following properties:
\begin{itemize}
\item[(a)] $\mathcal{M}_\mu$ is truthful and normalized.

\item[(b)] $\mathcal{M}_\mu$ is universally ex-post individually rational and ex-post \noPositiveTransfers. Moreover, given a bid vector $\bid$, it never pays any agent $i$ more than
	$\bid_i(o) (\frac{1}{\mu}-1)$, where $o=\mA(\bid)\in\outcomes$.

\item[(c)] For any bid vector $\bid\in \types$ (and any fixed random seed of nature) allocations $\tA(\bid)$ and $\mA(\bid)$ are identical with probability at least $1-n\mu$.

\item[(d)] If $\mA$ is $\alpha$-approximate (for social welfare) then $\tA$ is $\alpha/\left( 1-\tfrac{2}{1-\mu} \right)$-approximate.
\end{itemize}
\end{theorem}

\begin{proof}
Parts (b) and (c) follow immediately from Theorem~\ref{thm:maim-g}, and part (d) follows immediately from Theorem~\ref{thm:improved-approximation}. Thus, it remains to prove part (a).

Note that  the single-parameter allocation rule $\tA_{\bid}$ has the following property: for each agent $i$ the single-parameter bid $\lambda_i$ is rescaled by the (randomly chosen) factor $\chi_i\in[0,1]$ which does not depend on the bid, and then $\mA_{\bid}$ is called. Therefore, letting $\chi = (\chi_1\LDOTS \chi_n)$, it holds that
\begin{align}\label{eq:lm:multiParam-reduction-assumption}
    \tA_{\bid}(\lambda) = \mA( \chi\otimes (\lambda \otimes \bid) )\quad
        \text{ for all $\bid\in \types$ and $\lambda\in [0,1]^n$}.
\end{align}

We claim that $\tA$ is cycle-monotone. Indeed, fix bid vector $\bid\in\types$, agent $i$, some $k\geq 2$, and a $k$-tuple
 $\type_{i,0},\; \type_{i,1} \LDOTS \type_{i,k} \in \types_i$
of this agent's types. Let us consider a fixed realization of the random vector $\chi\in [0,1]^n$. For each type $\type_{i,j}$, note that (by Equations~\eqref{eq:lm:multiParam-reduction-assumption} and~\eqref{eq:multiParam-reduction}) we have
$$ \tA(\type_{i,j}, \bid_{-i})
    = \tA_{(\type_{i,j}, \bid_{-i})}(\,\vec{1}\,)
    = \mA\left( \chi \otimes (\bid_{-i},\,\type_{i,j}) \right) \in \outcomes.
$$
Denote this outcome by $o_{i,j}(\chi)$. Apply the cycle-monotonicity of $\mA$ for bid vector $\chi \otimes (\type_{i,j}, \bid_{-i})$:
\begin{align}\label{eq:proof-CMON}
\E_{\mA}\left[ \sum_{j=0}^k \type_{i,j}(o_{i,j}(\chi)) -
    \type_{i,\;(j-1) \bmod{k}}(o_{i,j}(\chi)) \right]\geq 0.
\end{align}
Recalling that
    $o_{i,j}(\chi) = \tA(\type_{i,j}, \bid_{-i})$,
we observe that for this fixed realization of $\chi$,~\refeq{eq:proof-CMON} is exactly the inequality in the definition of cycle-monotonicity for $\tA$. Therefore taking expectation over $\chi$, we obtain the desired inequality~\eqref{eq:CMON} for $\tA$. Claim proved.%
\footnote{Note that the proof of cycle-monotonicity of $\tA$ did not use any other property of the canonical self-resampling procedures $\mathbf{f}_{\mu}$ other than~\refeq{eq:lm:multiParam-reduction-assumption}. The truthfulness properties of $\mathbf{f}_{\mu}$ are used in the forthcoming argument about payments.}

\newcommand{\hA}{\widehat{\mA}}
\newcommand{\hP}{\widehat{\mP}}

It remains to prove that in the transformed mechanism $(\tA,\tP)$, the payment rule satisfies~\refeq{eq:multiParam-myerson}. Fix bid vector $\bid$ and consider the transformed single-parameter mechanism $(\tA_{\bid},\tP_{\bid})$ for the single-parameter type space $\types_{\bid}$. In the terminology of single-parameter domains, each agent $i$ receives an allocation
    $\tA_{\bid,\,i}(\lambda) = b_i(\tA_{\bid}(\lambda))$
whenever the bid vector is
    $\lambda \in [0,1]^n$.
Since this is a truthful and normalized single-parameter mechanism, it follows that
\begin{align*}
\E\left[ \tP_{\bid}(\lambda) \right]  =
\E\left[ \lambda_i\, \tA_{\bid,\,i}(\lambda) -
 \int_0^{\lambda_i}\,  \tA_{\bid,\,i}(\lambda_{-i}, u) \,du\right],
 \quad \forall \lambda\in [0,1]^n.
\end{align*}
Plugging in $\lambda = \vec{1}$ and~\refeq{eq:multiParam-reduction},
we obtain the desired~\refeq{eq:multiParam-myerson}.
\end{proof}


\section{Open Questions}
\label{sec:open-questions}

This paper gives rise to a number of open questions.  As discussed in Section~\ref{sec:followup}, some of these questions have been partially addressed in the follow-up work. Here we present the current status.

\xhdr{Variance vs. expectation tradeoff.}
Randomized mechanisms constructed via our general transformation exhibit an explicit tradeoff between the variance in payments and the loss in expected welfare compared to the optimal allocation rule. Since the variance in payments can be very high,
it is desirable to optimize this tradeoff (to complement the expectation-only guarantees).

The worst-case optimality result in \citeN{SingleCall-ec12}, discussed in Section~\ref{sec:followup}, does not resolve this question, since it does not rule out a reduction which achieves a better tradeoff for some (but not all) monotone allocation rules. Further, the optimal tradeoff for a given domain could be achieved by a mechanism that cannot be presented as a reduction from some welfare-optimal allocation rule.

Our informal conjecture is that the tradeoff in this paper is optimal for any given single-parameter domain with ``informational obstacle", i.e. whenever payment computation for welfare-optimal allocation rule is impossible
due to the insufficient observable information.

A specific formal conjecture is that our tradeoff is optimal for MAB mechanisms. To take an extreme version, what welfare loss can be achieved if no rebates (i.e., no positive transfers) are allowed?

\xhdr{The power of randomization.}
We have a separation result for randomized vs. deterministic ex-post truthful MAB mechanisms. Can one obtain similar separation results for other single-parameter domains? The positive side for any such hypothetical separation result is provided by our general reduction, so it remains to produce the corresponding negative result for deterministic mechanisms. However, such negative results are not likely to be easy, considering the difficulties faced by~\cite{MechMAB-ec09,DevanurK09} for MAB mechanisms. One specific target would be the router scheduling problem proposed in~\cite{Parkes-netecon12}.

\xhdr{MAB allocation rules.}
This paper opens up the problem of designing monotone MAB allocation rules, which is a new angle in the rich literature on MAB (also see~\cite{MonotoneMAB-colt11}). While we have focused on stochastic MAB, many other MAB settings have been studied in the literature, making various assumptions on payoff evolution over time (e.g.,~\cite{bandits-exp3,DynamicMAB-colt08,Hazan-soda09}), dependencies between arms (e.g.,~\cite{FlaxmanKM-soda05,yahoo-bandits-icml07,LipschitzMAB-stoc08,GPbandits-icml10}), and side information available to the algorithm (e.g.,~\cite{LipschitzMAB-stoc08,Langford-nips07,contextualMAB-colt11}).
\OMIT{a frequent theme is to take advantage of some ``benign" feature of the problem instance while maintaining the optimal ``worst-case" performance (e.g.,~\cite{LipschitzMAB-stoc08,Hazan-soda09,contextualMAB-colt11}).}
For most such settings one could meaningfully define the corresponding mechanism design problem; we have reduced this problem to that of designing monotone MAB allocation rules. In particular, for any given MAB setting one could ask whether monotone MAB allocation rules can achieve optimal regret.

One appealing target here is the adversarial MAB setting (with oblivious adversary). The ex-post truthful mechanism in~\cite{MechMAB-ec09} achieves regret $\tilde{O}(k^{1/3}\,T^{2/3})$ for this setting, whereas the best known MAB algorithms achieve regret $O(\sqrt{kT})$~\cite{bandits-exp3,Bubeck-colt09}; it is not clear what is the tight regret bound.

\xhdr{More applications.}
In addition to the applications presented in this paper and the follow-up work, what other domains can our general reduction (and the multi-parameter extension thereof) be fruitfully applied to? In particular, one could consider two generalizations of MAB mechanisms: to multiple ads per agent and to multiple ad slots with slot-dependent values-per-click.

\appendix

\section*{APPENDIX: $\SecondAlg$ is equivalent to $\FirstAlg$ (proof of Proposition~\ref{prop:srsp})}
\setcounter{section}{1}
\label{app:canon-equiv}

\OMIT{The following lemma provides a useful equivalent description of the
sampling distribution of $\SecondAlg$.}

\newcommand{\ThirdAlg}[1]{{\mbox{\sc bdr}_{#1}}}

Let us compare the sampling procedures defined by
$\FirstAlg$ and $\SecondAlg$.
To simplify the notation, we will omit the subscript $i$ from the description of the procedures. That is, a self-resampling procedure inputs a scalar bid $b$ and a random seed $w$, and outputs two numbers $(x,y)$.
To prove that $\SecondAlg$ is equivalent to $\FirstAlg$, we
analyze a family of sampling rules that uses bounded-depth
recursion to ``interpolate'' between $\SecondAlg$ and $\FirstAlg$.
Specifically,
define $\ThirdAlg{k}$ to be the following family
of sampling algorithms
parameterized by $k \in \mathbb{N} \cup \{\infty\}$,
where $k-1$ is interpreted as $\infty$ when $k=\infty$.
\begin{center}
\begin{algorithm}[h]
\caption{The sampling algorithm $\ThirdAlg{k}$: Bounded Depth Recursion.}
\label{alg:bd-recur}
\begin{algorithmic}[1]
\STATE {\bf Input:} bid $b\in [0,\infty]$, parameter $\mu\in(0,1)$.
\STATE {\bf Output:} $(x,y)$ such that $0\leq x\leq y \leq b$.
\vspace{1mm}
\STATE \WProb{$1-\mu$}
\STATE \MyTab $x \leftarrow b$, $y \leftarrow b$.
\STATE {\bf else}
\STATE \MyTab {\bf if} $k=0$
\STATE \MyTab \MyTab Pick $\gamma_1, \gamma_2 \in [0,1]$ indep., uniformly at random.
\STATE \MyTab \MyTab $x \leftarrow b \cdot \gamma_1^{1/(1-\mu)}$,~~
		 $y \leftarrow b \cdot \max\{ \gamma_1^{1/(1-\mu)}, \gamma_2^{1/\mu} \}$.
\STATE \MyTab {\bf else} \MyTab {\em $\sslash \; k>0$}
\STATE \MyTab \MyTab Pick $b' \in [0,b]$ uniformly at random.
\STATE \MyTab \MyTab $(x',y') = \ThirdAlg{k-1}(b',\mu)$.
\STATE \MyTab \MyTab $x \leftarrow x'$, $y \leftarrow b'$.
\end{algorithmic}
\end{algorithm}
\end{center}
The reader may easily verify that $\ThirdAlg{k}$ is equal to
$\SecondAlg$ when $k=0$ and that it is equal to
$\FirstAlg$ when $k=\infty$. Furthermore, for any
$k < k' $ (where $k'\leq \infty$) there is an obvious coupling of $\ThirdAlg{k}$
with $\ThirdAlg{k'}$ such that the
two algorithms have probability at most
$\mu^{k+1}$ of  outputting different results:
simply let the two executions share the
same randomness until the recursion depth equals $k$.
Thus, as $k \to \infty$, the output distribution of $\ThirdAlg{k}$ converges,
in total variation distance, to that of $\ThirdAlg{\infty}$.
We will prove that for every finite $k$ the algorithms
$\ThirdAlg{k}$ and $\ThirdAlg{k+1}$ have identical
output distributions, from which it follows that their
output distribution is identical to that of $\ThirdAlg{0}$
and, therefore, that $\ThirdAlg{\infty}$ also has the same
output distribution as $\ThirdAlg{0}$, confirming
Proposition~\ref{prop:srsp}.

Couple $\ThirdAlg{k}$ and $\ThirdAlg{k+1}$ so
that they use shared randomness until the two algorithms
reach differing points in their control flow. This occurs
when the first algorithm is executing a call to
$\ThirdAlg{0}$ and the second algorithm is executing
a call to $\ThirdAlg{1}$ on the same input $(\beta,\mu)$.
(We are denoting the input in this step of the recursive
algorithms by $(\beta,\mu)$ rather
than $(b,\mu)$, to distinguish $\beta$ from the value
of $b$ on which the two algorithms $\ThirdAlg{k},
\ThirdAlg{k+1}$ were originally called.) At this point,
with probability $1-\mu$
both algorithms output $(x,y)=(\beta,\beta)$.
Conditional on this event \emph{not} taking place,
$\ThirdAlg{0}$ outputs $(x,y) = (\gamma_1^{1/p} \beta,
\max \{\gamma_1^{1/p},\gamma_2^{1/q}\} \beta)$
where $p = 1-\mu, \, q = \mu$.
Instead $\ThirdAlg{1}$ computes $\beta' = \gamma_3 \beta$
where $\gamma_3 \in [0,1]$ is uniformly random,
and it outputs $(x,y) = (\beta',\beta')$ with probability $p$
and otherwise $(x,y) = (\gamma_1^{1/p} \beta', \beta')$.
Lemma~\ref{lem:oneshot-equiv} tells us that these two
output distributions are the same.

\begin{lemma} \label{lem:oneshot-equiv}
Let $\gamma_1,\gamma_2,\gamma_3$ be mutually independent
random variables, each uniformly distributed in $[0,1]$.
Let $p,q > 0$ be numbers such that $p + q = 1$. Define
random variables $x,y,z$ by:
\begin{align*}
x &= \gamma_1^{1/p} \\
y &= \max\{\gamma_1^{1/p}, \gamma_2^{1/q}\} \\
z &= \begin{cases}
\gamma_3 & \mbox{if $\gamma_2 < p$} \\
\gamma_1^{1/p} \gamma_3 & \mbox{if $\gamma_2 \geq p$}
\end{cases}
\end{align*}
Then the pairs $(x,y)$ and $(z,\gamma_3)$ are identically
distributed.
\end{lemma}
\begin{proof}
We will show, equivalently, that
the pairs $(y,x/y)$ and $(\gamma_3,z/\gamma_3)$
are identically distributed. The distribution of
$(\gamma_3,z/\gamma_3)$ is completely characterized
by the following facts which are immediate from the
definition of $z$.
\begin{enumerate}
\item $\gamma_3$ and $z/\gamma_3$ are independent;
\item $\gamma_3$ is uniformly distributed in $[0,1]$;
\item $z/\gamma_3$ is equal to $1$ with probability $p$,
and conditional on $z/\gamma_3 \neq 1$, the distribution
of $(z/\gamma_3)^p$ is uniform on $[0,1)$.
\end{enumerate}
To finish the proof of the lemma, we shall prove the
corresponding facts about $y$ and $x/y$. Let $I_1,I_2 \subseteq [0,1]$
be any pair of intervals (open, closed, or half-open).
Let $a,b$ be the endpoints of $I_1$ and $c,d$ the endpoints
of $I_2$. To compute
$\Pr(y \in I_1, \; x/y \in I_2)$ it suffices to make the following
two observations:
\begin{align*}
\Pr(y \in I_1, \; x/y=1) &= \Pr
\left( a \leq \gamma_1^{1/p} \leq b, \;
0 \leq \gamma_2^{1/q} \leq \gamma_1^{1/p} \right) \\
&= \Pr \left( a^p \leq \gamma_1 \leq b^p, \;
0 \leq \gamma_2 \leq \gamma_1^{q/p} \right) \\
&= \int_{a^{p}}^{b^{p}} t^{q/p} \, dt = \int_{a^p}^{b^p} t^{1/p - 1} \, dt
= p (b - a) \\[2em]
\Pr(y \in I_1, \; x/y \in I_2 \setminus \{1\}) &=
\Pr(a \leq \gamma_2^{1/q} \leq b, \;
c \gamma_2^{1/q} \leq \gamma_1^{1/p} < d \gamma_2^{1/q}) \\
&=
\Pr(a^q \leq \gamma_2 \leq b^q, \;
c^p \gamma_2^{p/q} \leq \gamma_1 \leq d^p \gamma_2^{p/q}) \\
&= \int_{a^q}^{b^q} (d^p - c^p) t^{p/q} \, dt
= (d^p-c^p) \int_{a^q}^{b^q} t^{1/q - 1} \, dt
= q (d^p-c^p) (b-a).
\end{align*}
Therefore,
\[
\Pr(y \in I_1, \; x/y \in I_2) =
(b-a) \cdot \begin{cases}
q (d^p - c^p) & \mbox{if $1 \not\in I_2$} \\
p + q (d^p - c^p) & \mbox{if $1 \in I_2$}
\end{cases}.
\]
From this formula it follows that
$y$ and $x/y$ are independent,
$y$ is uniformly distributed, $\Pr(x/y=1)=p$,
and the distribution of $(x/y)^p$ conditional
on $x/y \neq 1$ is uniform on $[0,1)$.
\end{proof}

\section*{Acknowledgments}
We are indebted to Tim Roughgarden for suggesting that for positive types, an improved bound on the social welfare is possible (see Section~\ref{sec:welfare}). We would like to acknowledge that the preliminary form of a generalization of our ``generic transformation" to negative bids, and (the preliminary form of) the applications to offline mechanism design, have been derived jointly with Jason Hartline.

\bibliographystyle{ACM-Reference-Format-Journals}
\bibliography{bib-abbrv,bib-bandits,bib-slivkins,bib-AGT}


\begin{thebibliography}{00}


\ifx \showCODEN    \undefined \def \showCODEN     #1{\unskip}     \fi
\ifx \showDOI      \undefined \def \showDOI       #1{{\tt DOI:}\penalty0{#1}\ }
  \fi
\ifx \showISBNx    \undefined \def \showISBNx     #1{\unskip}     \fi
\ifx \showISBNxiii \undefined \def \showISBNxiii  #1{\unskip}     \fi
\ifx \showISSN     \undefined \def \showISSN      #1{\unskip}     \fi
\ifx \showLCCN     \undefined \def \showLCCN      #1{\unskip}     \fi
\ifx \shownote     \undefined \def \shownote      #1{#1}          \fi
\ifx \showarticletitle \undefined \def \showarticletitle #1{#1}   \fi
\ifx \showURL      \undefined \def \showURL       #1{#1}          \fi

\bibitem[\protect\citeauthoryear{Archer and Kleinberg}{Archer and
  Kleinberg}{2008a}]%
        {ArcherKleinberg-sigecom08}
{Aaron Archer} {and} {Robert Kleinberg}. 2008a.
\newblock \showarticletitle{Characterizing truthful mechanisms with convex type
  spaces}.
\newblock {\em SIGecom Exchanges\/} {7}, 3 (2008).
\newblock


\bibitem[\protect\citeauthoryear{Archer and Kleinberg}{Archer and
  Kleinberg}{2008b}]%
        {ArcherKleinberg-ec08}
{Aaron Archer} {and} {Robert Kleinberg}. 2008b.
\newblock \showarticletitle{Truthful germs are contagious: a local to global
  characterization of truthfulness}. In {\em 9th ACM Conf. on Electronic
  Commerce (EC)}. 21--30.
\newblock


\bibitem[\protect\citeauthoryear{Archer, Papadimitriou, Talwar, and
  Tardos}{Archer et~al\mbox{.}}{2004}]%
        {APTT04}
{Aaron Archer}, {Christos Papadimitriou}, {Kunal Talwar}, {and} {{\'E}va
  Tardos}. 2004.
\newblock \showarticletitle{An approximate truthful mechanism for combinatorial
  auctions with single parameter agents}.
\newblock {\em Internet Mathematics\/}  {1} (2004), 129--150.
\newblock
\newblock
\shownote{Extended abstract in \emph{SODA 2003}.}


\bibitem[\protect\citeauthoryear{Archer and Tardos}{Archer and Tardos}{2001}]%
        {ArcherTardos}
{Aaron Archer} {and} {{\'E}va Tardos}. 2001.
\newblock \showarticletitle{Truthful Mechanisms for One-Parameter Agents}. In
  {\em IEEE Symp. on Foundations of Computer Science (FOCS)}. 482--491.
\newblock


\bibitem[\protect\citeauthoryear{Ashlagi, Braverman, Hassidim, and
  Monderer}{Ashlagi et~al\mbox{.}}{2010}]%
        {Ashlagi-econometrica10}
{Itai Ashlagi}, {Mark Braverman}, {Avinatan Hassidim}, {and} {Dov Monderer}.
  2010.
\newblock \showarticletitle{Monotonicity and Implementability}.
\newblock {\em Econometrica\/} {78}, 5 (2010), 1749--1772.
\newblock


\bibitem[\protect\citeauthoryear{Athey and Segal}{Athey and Segal}{2013}]%
        {AtheySegal-econometrica13}
{Susan Athey} {and} {Ilya Segal}. 2013.
\newblock \showarticletitle{An Efficient Dynamic Mechanism}.
\newblock {\em Econometrica\/} {81}, 6 (Nov. 2013), 2463--2485.
\newblock
\newblock
\shownote{A preliminary version has been available as a working paper since
  2007.}


\bibitem[\protect\citeauthoryear{Audibert and Bubeck}{Audibert and
  Bubeck}{2010}]%
        {Bubeck-colt09}
{J.Y. Audibert} {and} {S. Bubeck}. 2010.
\newblock \showarticletitle{{Regret Bounds and Minimax Policies under Partial
  Monitoring}}.
\newblock {\em J. of Machine Learning Research (JMLR)\/}  {11} (2010),
  2785--2836.
\newblock
\newblock
\shownote{A preliminary version has been published in \emph{COLT 2009}.}


\bibitem[\protect\citeauthoryear{Auer, Cesa-Bianchi, and Fischer}{Auer
  et~al\mbox{.}}{2002a}]%
        {bandits-ucb1}
{Peter Auer}, {Nicol{\`o} Cesa-Bianchi}, {and} {Paul Fischer}. 2002a.
\newblock \showarticletitle{Finite-time Analysis of the Multiarmed Bandit
  Problem.}
\newblock {\em Machine Learning\/} {47}, 2-3 (2002), 235--256.
\newblock
\newblock
\shownote{Preliminary version in {\em 15th ICML}, 1998.}


\bibitem[\protect\citeauthoryear{Auer, Cesa-Bianchi, Freund, and Schapire}{Auer
  et~al\mbox{.}}{2002b}]%
        {bandits-exp3}
{Peter Auer}, {Nicol{\`o} Cesa-Bianchi}, {Yoav Freund}, {and} {Robert~E.
  Schapire}. 2002b.
\newblock \showarticletitle{The Nonstochastic Multiarmed Bandit Problem.}
\newblock {\em SIAM J. Comput.\/} {32}, 1 (2002), 48--77.
\newblock
\newblock
\shownote{Preliminary version in {\em 36th IEEE FOCS}, 1995.}


\bibitem[\protect\citeauthoryear{Babaioff, Blumrosen, and Schapira}{Babaioff
  et~al\mbox{.}}{2013}]%
        {BabaioffBS13}
{Moshe Babaioff}, {Liad Blumrosen}, {and} {Michael Schapira}. 2013.
\newblock \showarticletitle{The communication burden of payment determination}.
\newblock {\em Games and Economic Behavior\/} {77}, 1 (2013), 153 -- 167.
\newblock
\newblock
\shownote{Preliminary version in {ACM EC 2008}.}


\bibitem[\protect\citeauthoryear{Babaioff, Kleinberg, and Slivkins}{Babaioff
  et~al\mbox{.}}{2010}]%
        {Transform-ec10-conf}
{Moshe Babaioff}, {Robert Kleinberg}, {and} {Aleksandrs Slivkins}. 2010.
\newblock \showarticletitle{Truthful Mechanisms with Implicit Payment
  Computation}. In {\em 11th ACM Conf. on Electronic Commerce (EC)}. 43--52.
\newblock


\bibitem[\protect\citeauthoryear{Babaioff, Kleinberg, and Slivkins}{Babaioff
  et~al\mbox{.}}{2013}]%
        {BKS2-ec13}
{Moshe Babaioff}, {Robert Kleinberg}, {and} {Aleksandrs Slivkins}. 2013.
\newblock \showarticletitle{Multi-parameter mechanisms with implicit payment
  computation}. In {\em 13th ACM Conf. on Electronic Commerce (EC)}. 35--52.
\newblock


\bibitem[\protect\citeauthoryear{Babaioff, Sharma, and Slivkins}{Babaioff
  et~al\mbox{.}}{2014}]%
        {MechMAB-ec09}
{Moshe Babaioff}, {Yogeshwer Sharma}, {and} {Aleksandrs Slivkins}. 2014.
\newblock \showarticletitle{Characterizing Truthful Multi-armed Bandit
  Mechanisms}.
\newblock {\em SIAM J. on Computing (SICOMP)\/} {43}, 1 (2014), 194--230.
\newblock
\newblock
\shownote{Preliminary version in \emph{10th ACM EC}, 2009.}


\bibitem[\protect\citeauthoryear{Bei and Huang}{Bei and Huang}{2011}]%
        {BH11}
{Xiaohui Bei} {and} {Zhiyi Huang}. 2011.
\newblock \showarticletitle{Bayesian Incentive Compatibility via Fractional
  Assignments}. In {\em 22nd ACM-SIAM Symp. on Discrete Algorithms (SODA)}.
  720--733.
\newblock


\bibitem[\protect\citeauthoryear{Bergemann and Said}{Bergemann and
  Said}{2011}]%
        {DynAuctions-survey10}
{Dirk Bergemann} {and} {Maher Said}. 2011.
\newblock \showarticletitle{Dynamic Auctions: A Survey}.
\newblock In {\em Wiley Encyclopedia of Operations Research and Management
  Science, Vol. 2}. Wiley: New York, 1511--1522.
\newblock


\bibitem[\protect\citeauthoryear{Bergemann and V\"{a}lim\"{a}ki}{Bergemann and
  V\"{a}lim\"{a}ki}{2010}]%
        {DynPivot-econometrica10}
{Dirk Bergemann} {and} {Juuso V\"{a}lim\"{a}ki}. 2010.
\newblock \showarticletitle{The Dynamic Pivot Mechanism}.
\newblock {\em Econometrica\/} {78}, 2 (2010), 771--789.
\newblock
\newblock
\shownote{Preliminary versions have been available since 2006, as \emph{Cowles
  Foundation Discussion Papers} \#1584 (2006), \#1616 (2007) and \#1672(2008).}


\bibitem[\protect\citeauthoryear{Briest, Chawla, Kleinberg, and
  Weinberg}{Briest et~al\mbox{.}}{2014}]%
        {j-soda10-bries}
{Patrick Briest}, {Shuchi Chawla}, {Robert Kleinberg}, {and} {S.~Matthew
  Weinberg}. 2014.
\newblock \showarticletitle{Pricing Lotteries}.
\newblock {\em Journal of Economic Theory\/} (2014).
\newblock
\newblock
\shownote{In press, accepted manuscript.}


\bibitem[\protect\citeauthoryear{Cai, Daskalakis, and Weinberg}{Cai
  et~al\mbox{.}}{2012}]%
        {CDW12}
{Yang Cai}, {Constantinos Daskalakis}, {and} {S.~Matthew Weinberg}. 2012.
\newblock \showarticletitle{Optimal Multi-dimensional Mechanism Design:
  Reducing Revenue to Welfare Maximization}. In {\em 53th IEEE Symp. on
  Foundations of Computer Science (FOCS)}. 130--139.
\newblock


\bibitem[\protect\citeauthoryear{Cai, Daskalakis, and Weinberg}{Cai
  et~al\mbox{.}}{2013a}]%
        {CDW13a}
{Yang Cai}, {Constantinos Daskalakis}, {and} {S.~Matthew Weinberg}. 2013a.
\newblock \showarticletitle{Reducing Revenue to Welfare Maximization:
  Approximation Algorithms and other Generalizations}. In {\em 24nd ACM-SIAM
  Symp. on Discrete Algorithms (SODA)}. 578--595.
\newblock


\bibitem[\protect\citeauthoryear{Cai, Daskalakis, and Weinberg}{Cai
  et~al\mbox{.}}{2013b}]%
        {CDW13b}
{Yang Cai}, {Constantinos Daskalakis}, {and} {S.~Matthew Weinberg}. 2013b.
\newblock \showarticletitle{Understanding Incentives: Mechanism Design Becomes
  Algorithm Design}. In {\em 54th IEEE Symp. on Foundations of Computer Science
  (FOCS)}. 618--627.
\newblock


\bibitem[\protect\citeauthoryear{Chawla, Immorlica, and Lucier}{Chawla
  et~al\mbox{.}}{2012}]%
        {CIL12}
{Shuchi Chawla}, {Nicole Immorlica}, {and} {Brendan Lucier}. 2012.
\newblock \showarticletitle{On the limits of black-box reductions in mechanism
  design}. In {\em 44th ACM Symp. on Theory of Computing (STOC)}. 435--448.
\newblock


\bibitem[\protect\citeauthoryear{Daskalakis and Weinberg}{Daskalakis and
  Weinberg}{2014}]%
        {DW14}
{Constantinos Daskalakis} {and} {S.~Matthew Weinberg}. 2014.
\newblock Bayesian Truthful Mechanisms for Job Scheduling from Bi-criterion
  Approximation Algorithms.
\newblock   (2014).
\newblock
\newblock
\shownote{arXiv:1405.5940.}


\bibitem[\protect\citeauthoryear{Devanur and Kakade}{Devanur and
  Kakade}{2009}]%
        {DevanurK09}
{Nikhil Devanur} {and} {Sham~M. Kakade}. 2009.
\newblock \showarticletitle{The Price of Truthfulness for Pay-Per-Click
  Auctions}. In {\em 10th ACM Conf. on Electronic Commerce (EC)}. 99--106.
\newblock


\bibitem[\protect\citeauthoryear{Dobzinski and Dughmi}{Dobzinski and
  Dughmi}{2009}]%
        {DD09}
{Shahar Dobzinski} {and} {Shaddin Dughmi}. 2009.
\newblock \showarticletitle{On the Power of Randomization in Algorithmic
  Mechanism Design}. In {\em 50th IEEE Symp. on Foundations of Computer Science
  (FOCS)}.
\newblock


\bibitem[\protect\citeauthoryear{Dobzinski, Fu, and Kleinberg}{Dobzinski
  et~al\mbox{.}}{2012}]%
        {j-stoc11-dobzi}
{Shahar Dobzinski}, {Hu Fu}, {and} {Robert~D. Kleinberg}. 2012.
\newblock \showarticletitle{Optimal auctions with correlated bidders are easy}.
\newblock {\em Games and Economic Behavior\/} (2012).
\newblock
\newblock
\shownote{Special issue on selected algorithmic game theory papers from STOC,
  FOCS, and SODA 2011. In press.}


\bibitem[\protect\citeauthoryear{Flaxman, Kalai, and McMahan}{Flaxman
  et~al\mbox{.}}{2005}]%
        {FlaxmanKM-soda05}
{Abraham Flaxman}, {Adam Kalai}, {and} {H.~Brendan McMahan}. 2005.
\newblock \showarticletitle{{Online Convex Optimization in the Bandit Setting:
  Gradient Descent without a Gradient}}. In {\em 16th ACM-SIAM Symp. on
  Discrete Algorithms (SODA)}. 385--394.
\newblock


\bibitem[\protect\citeauthoryear{Gatti, Lazaric, and Trovo}{Gatti
  et~al\mbox{.}}{2012}]%
        {Gatti-ec12}
{Nicola Gatti}, {Alessandro Lazaric}, {and} {Francesco Trovo}. 2012.
\newblock \showarticletitle{{A Truthful Learning Mechanism for Contextual
  Multi-Slot Sponsored Search Auctions with Externalities}}. In {\em 13th ACM
  Conf. on Electronic Commerce (EC)}.
\newblock


\bibitem[\protect\citeauthoryear{Hartline}{Hartline}{2012}]%
        {Hartline-book-draft}
{Jason~D. Hartline}. 2012.
\newblock {Approximation in Economic Design}.
\newblock Draft of a forthcoming book, available at {\tt
  http://jasonhartline.com/MDnA/}.   (2012).
\newblock


\bibitem[\protect\citeauthoryear{Hartline, Kleinberg, and Malekian}{Hartline
  et~al\mbox{.}}{2011}]%
        {HKM11}
{Jason~D. Hartline}, {Robert Kleinberg}, {and} {Azarakhsh Malekian}. 2011.
\newblock \showarticletitle{Bayesian Incentive Compatibility via Matchings}. In
  {\em 22nd ACM-SIAM Symp. on Discrete Algorithms (SODA)}. 734--747.
\newblock


\bibitem[\protect\citeauthoryear{Hartline and Lucier}{Hartline and
  Lucier}{2010}]%
        {HL10}
{Jason~D. Hartline} {and} {Brendan Lucier}. 2010.
\newblock \showarticletitle{Bayesian algorithmic mechanism design}. In {\em
  42th ACM Symp. on Theory of Computing (STOC)}. 301--310.
\newblock


\bibitem[\protect\citeauthoryear{Hazan and Kale}{Hazan and Kale}{2009}]%
        {Hazan-soda09}
{Elad Hazan} {and} {Satyen Kale}. 2009.
\newblock \showarticletitle{{Better algorithms for benign bandits}}. In {\em
  20th ACM-SIAM Symp. on Discrete Algorithms (SODA)}. 38--47.
\newblock


\bibitem[\protect\citeauthoryear{Huang and Kannan}{Huang and Kannan}{2012}]%
        {Huang-ExpMech12}
{Zhiyi Huang} {and} {Sampath Kannan}. 2012.
\newblock {Exponential mechanism for social welfare: private, truthful, and
  nearly optimal}.
\newblock Working paper.   (2012).
\newblock


\bibitem[\protect\citeauthoryear{Jain, Menache, Naor, and Yaniv}{Jain
  et~al\mbox{.}}{2011}]%
        {Jain-sagt11}
{Navendu Jain}, {Ishai Menache}, {Joseph Naor}, {and} {Jonathan Yaniv}. 2011.
\newblock \showarticletitle{{A Truthful Mechanism for Value-Based Scheduling in
  Cloud Computing}}. In {\em 4th Symp. on Algorithmic Game Theory (SAGT)}.
\newblock


\bibitem[\protect\citeauthoryear{Kleinberg, Niculescu-Mizil, and
  Sharma}{Kleinberg et~al\mbox{.}}{2008a}]%
        {sleeping-colt08}
{Robert Kleinberg}, {Alexandru Niculescu-Mizil}, {and} {Yogeshwer Sharma}.
  2008a.
\newblock \showarticletitle{Regret bounds for sleeping experts and bandits}. In
  {\em 21st Conf. on Learning Theory (COLT)}. 425--436.
\newblock


\bibitem[\protect\citeauthoryear{Kleinberg, Slivkins, and Upfal}{Kleinberg
  et~al\mbox{.}}{2008b}]%
        {LipschitzMAB-stoc08}
{Robert Kleinberg}, {Aleksandrs Slivkins}, {and} {Eli Upfal}. 2008b.
\newblock \showarticletitle{{Multi-Armed Bandits in Metric Spaces}}. In {\em
  40th ACM Symp. on Theory of Computing (STOC)}. 681--690.
\newblock


\bibitem[\protect\citeauthoryear{Lai and Robbins}{Lai and Robbins}{1985}]%
        {Lai-Robbins-85}
{T.L. Lai} {and} {Herbert Robbins}. 1985.
\newblock \showarticletitle{{Asymptotically efficient Adaptive Allocation
  Rules}}.
\newblock {\em Advances in Applied Mathematics\/}  {6} (1985), 4--22.
\newblock


\bibitem[\protect\citeauthoryear{Langford and Zhang}{Langford and
  Zhang}{2007}]%
        {Langford-nips07}
{John Langford} {and} {Tong Zhang}. 2007.
\newblock \showarticletitle{{The Epoch-Greedy Algorithm for Contextual
  Multi-armed Bandits}}. In {\em 21st Advances in Neural Information Processing
  Systems (NIPS)}.
\newblock


\bibitem[\protect\citeauthoryear{Manelli and Vincent}{Manelli and
  Vincent}{2006}]%
        {MV}
{Alejandro~M. Manelli} {and} {Daniel~R. Vincent}. 2006.
\newblock \showarticletitle{Bundling as an optimal selling mechanism for a
  multiple-good monopolist}.
\newblock {\em Journal of Economic Theory\/} {127}, 1 (2006), 1 -- 35.
\newblock


\bibitem[\protect\citeauthoryear{McSherry and Talwar}{McSherry and
  Talwar}{2007}]%
        {McSherry-Talwar-focs07}
{Frank McSherry} {and} {Kunal Talwar}. 2007.
\newblock \showarticletitle{Mechanism Design via Differential Privacy}. In {\em
  48th IEEE Symp. on Foundations of Computer Science (FOCS)}. 94--103.
\newblock


\bibitem[\protect\citeauthoryear{Myerson}{Myerson}{1981}]%
        {Myerson}
{Roger~B. Myerson}. 1981.
\newblock \showarticletitle{{Optimal Auction Design}}.
\newblock {\em Mathematics of Operations Research\/}  {6} (1981), 58--73.
\newblock


\bibitem[\protect\citeauthoryear{Nisan and Ronen}{Nisan and Ronen}{2001}]%
        {NR01}
{N. Nisan} {and} {A. Ronen}. 2001.
\newblock \showarticletitle{{Algorithmic Mechanism Design}}.
\newblock {\em Games and Economic Behavior\/} {35}, 1-2 (2001), 166--196.
\newblock


\bibitem[\protect\citeauthoryear{Pandey, Chakrabarti, and Agarwal}{Pandey
  et~al\mbox{.}}{2007}]%
        {yahoo-bandits-icml07}
{Sandeep Pandey}, {Deepayan Chakrabarti}, {and} {Deepak Agarwal}. 2007.
\newblock \showarticletitle{{Multi-armed Bandit Problems with Dependent Arms}}.
  In {\em 24th Intl. Conf. on Machine Learning (ICML)}.
\newblock


\bibitem[\protect\citeauthoryear{Rochet}{Rochet}{1987}]%
        {Rochet-1987}
{Jean-Charles Rochet}. 1987.
\newblock \showarticletitle{A necessary and sufficient condition for
  rationalizability in a quasi-linear context}.
\newblock {\em J. of Mathematical Economics\/} {16}, 2 (April 1987),
  191--–200.
\newblock


\bibitem[\protect\citeauthoryear{Segal}{Segal}{2010}]%
        {Segal-personal-10}
{Ilya Segal}. 2010.
\newblock Personal communication.   (2010).
\newblock


\bibitem[\protect\citeauthoryear{Shnayder, Hoon, Parkes, and Kawadia}{Shnayder
  et~al\mbox{.}}{2012}]%
        {Parkes-netecon12}
{Victor Shnayder}, {Jeremy Hoon}, {David Parkes}, {and} {Vikas Kawadia}. 2012.
\newblock \showarticletitle{{Truthful Prioritization Schemes for Spectrum
  Sharing}}. In {\em 7th Workshop on the Economics of Networks, Systems and
  Computation (NetEcon)}.
\newblock


\bibitem[\protect\citeauthoryear{Slivkins}{Slivkins}{2011a}]%
        {contextualMAB-colt11}
{Aleksandrs Slivkins}. 2011a.
\newblock \showarticletitle{{Contextual Bandits with Similarity Information}}.
  In {\em 24th Conf. on Learning Theory (COLT)}.
\newblock
\newblock
\shownote{To appear in J. of Machine Learning Research (JMLR), 2014.}


\bibitem[\protect\citeauthoryear{Slivkins}{Slivkins}{2011b}]%
        {MonotoneMAB-colt11}
{Aleksandrs Slivkins}. 2011b.
\newblock {Monotone multi-armed bandit allocations}.
\newblock Open Problem Session at {\em COLT 2011} (Conf. on Learning Theory).
  (2011).
\newblock


\bibitem[\protect\citeauthoryear{Slivkins and Upfal}{Slivkins and
  Upfal}{2008}]%
        {DynamicMAB-colt08}
{Aleksandrs Slivkins} {and} {Eli Upfal}. 2008.
\newblock \showarticletitle{{Adapting to a Changing Environment: the Brownian
  Restless Bandits}}. In {\em 21st Conf. on Learning Theory (COLT)}. 343--354.
\newblock


\bibitem[\protect\citeauthoryear{Srinivas, Krause, Kakade, and Seeger}{Srinivas
  et~al\mbox{.}}{2010}]%
        {GPbandits-icml10}
{Niranjan Srinivas}, {Andreas Krause}, {Sham Kakade}, {and} {Matthias Seeger}.
  2010.
\newblock \showarticletitle{{Gaussian Process Optimization in the Bandit
  Setting: No Regret and Experimental Design}}. In {\em 27th Intl. Conf. on
  Machine Learning (ICML)}. 1015--1022.
\newblock


\bibitem[\protect\citeauthoryear{Thanassoulis}{Thanassoulis}{2004}]%
        {Than}
{John Thanassoulis}. 2004.
\newblock \showarticletitle{Haggling over substitutes}.
\newblock {\em J. Economic Theory\/}  {117} (2004), 217--245.
\newblock


\bibitem[\protect\citeauthoryear{Wilkens and Sivan}{Wilkens and Sivan}{2012}]%
        {SingleCall-ec12}
{Chris Wilkens} {and} {Balasubramanian Sivan}. 2012.
\newblock \showarticletitle{Single-Call Mechanisms}. In {\em 13th ACM Conf. on
  Electronic Commerce (EC)}.
\newblock


\end{thebibliography}

\end{document}